%% file: paper.tex
\documentclass[11pt,a4paper]{article}
\usepackage{style}
\usepackage{shortcuts}

\graphicspath{{./figures/} }

\newcommand\todo[1]{}

\title{Faster 0-1-Knapsack via Near-Convex Min-Plus-Convolution}
\author{%
    Karl Bringmann\footnote{Saarland University and Max Planck Institute for Informatics, Saarland Informatics Campus. This work is part of the project TIPEA that has received funding from the European Research Council (ERC) under the European Unions Horizon 2020 research and innovation programme (grant agreement No.~850979). Email: \url{bringmann@cs.uni-saarland.de}.}\and%
    Alejandro Cassis\footnote{Saarland University and Max Planck Institute for Informatics, Saarland Informatics Campus. This work is part of the project TIPEA that has received funding from the European Research Council (ERC) under the European Unions Horizon 2020 research and innovation programme (grant agreement No.~850979). Email: \url{acassis@cs.uni-saarland.de}.}}

\begin{document}
\maketitle

\begin{abstract}
We revisit the classic 0-1-Knapsack problem, in which we are given $n$ items with their weights and profits as well as a weight budget $W$, and the goal is to find a subset of items of total weight at most~$W$ that maximizes the total profit.
We study pseudopolynomial-time algorithms parameterized by the largest profit of any item $p_{\max}$, and the largest weight of any item $w_{\max}$.
Our main result are algorithms for 0-1-Knapsack running in time $\widetilde{O}(n\,\wmax\,\pmax^{2/3})$ and $\widetilde{O}(n\,\pmax\,\wmax^{2/3})$, improving upon an algorithm in time $O(n\,\pmax\,\wmax)$ by Pisinger [J. Algorithms '99].
In the regime $\pmax \approx \wmax \approx n$ (and $W \approx \OPT \approx n^2$) our algorithms are the first to break the cubic barrier $n^3$.

To obtain our result, we give an efficient algorithm to compute the min-plus convolution of \emph{near-convex} functions.
More precisely, we say that a function $f \colon [n] \mapsto \Int$ is $\Delta$-near convex with $\Delta \geq 1$, if there is a convex function $\breve{f}$ such that $\breve{f}(i) \leq f(i) \leq \breve{f}(i) + \Delta$ for every $i$.
We design an algorithm computing the min-plus convolution of two $\Delta$-near convex functions in time $\widetilde{O}(n\Delta)$.
This tool can replace the usage of the \emph{prediction technique} of Bateni, Hajiaghayi, Seddighin and Stein [STOC '18] in all applications we are aware of, and we believe it has wider applicability.

\end{abstract}

\thispagestyle{empty}
\setcounter{page}{0}

\newpage

\input{sections/introduction}
\input{sections/preliminaries}
\input{sections/knapsack-algorithm}
\input{sections/minplus-algorithm}

\bibliographystyle{halpha}
\bibliography{refs}

\end{document}

%% file: sections/introduction.tex
\section{Introduction}

In the 0-1-Knapsack problem, we are given a set of $n$ items $\calI = \set{(p_1,w_1),\dots,(p_n,w_n)}$, where item $i$ has a profit $p_i \in \Nat$ and a weight $w_i \in \Nat$, as well as a weight budget $W \in \Nat$.
The goal is to compute $\OPT := \max \sum_{i=1}^n p_i x_i$ subject to the contraints $\sum_{i=1}^n w_i x_i \leq W$ and $x \in \set{0,1}^n$.
This classic and fundamental problem in computer science and operations research has been studied for decades (see e.g. \cite{KellererPP04} for a book on the topic and related problems).
Knapsack is weakly NP-hard, and the textbook dynamic programming algorithm due to Bellman~\cite{Bellman57} solves it in time $O(n \cdot \min\set{W, \OPT})$.

Recent works have studied the fine-grained complexity of Knapsack and related problems, where the goal is to give best-possible pseudopolynomial-time algorithms with respect to different parameters, see~\cref{tab:running-times} and~\cite{Bringmann17,JinW19,BringmannW21,AxiotisBBJNTW21,ChanH22,DengMZ23,Klein22,EisenbrandW18,JansenR19}.
In this work we study the complexity of 0-1-Knapsack in terms of two natural parameters: the largest weight among the items denoted by $\wmax$, and the largest profit denoted by $\pmax$.
Note that we can assume without loss of generality that $\wmax \le W$ and $\pmax \le \OPT$.
Therefore, a small polynomial dependence on these parameters can lead to faster algorithms compared to the standard dynamic programming algorithm on certain instances.

This parameterization has been studied by several previous works, see~\cref{tab:running-times}.
To compare these running times, note that since any feasible solution includes at most all items, we can assume without loss of generality that $W \leq n\cdot\wmax$ and $\OPT \leq n\cdot \pmax$.
Note that when $\pmax \approx \wmax \approx n$ (and $W \approx \OPT \approx n^2$), all known algorithms require time $\Omega(n^3)$.
In particular, in this regime the algorithm in time $O(n\,\wmax\,\pmax)$ of Pisinger from '99~\cite{Pisinger99} is still the best known.
In this paper we \emph{overcome this cubic barrier}:

\begin{theorem}\label{thm:alg-knapsack}
	There is a randomized algorithm for 0-1-Knapsack that runs in time\footnote{We use $\Ohtilde(\cdot)$ to supress polylogarithmic factors in the input size and the largest input number.} \[\Ohtilde((\pmax W)^{2/3} (n \wmax)^{1/3} + n\wmax)\] and succeeds with high probability.
    Using the bound $W \leq n\wmax$, this running time is at most $\Ohtilde(n\,\wmax\,\pmax^{2/3})$.
\end{theorem}

Symmetrically, we obtain the following:

\begin{theorem}\label{thm:alg-knapsack-symmetric}
	There is a randomized algorithm for 0-1-Knapsack that runs in time \[\Ohtilde((\wmax \OPT)^{2/3} (n \pmax)^{1/3} + n\pmax)\] and succeeds with high probability.
    Using the bound $\OPT \leq n\pmax$, this running time is at most $\Ohtilde(n\,\pmax\,\wmax^{2/3})$.
\end{theorem}

\setlength{\tabcolsep}{10pt}
\bgroup
\def\arraystretch{1.3}
\begin{table}[ht]
    \caption{Pseudopolynomial-time algorithms for 0-1 Knapsack.} \label{tab:running-times}
    \centering
    \begin{tabular}{|ll|}
        \hline
        \thead{Reference} & \thead{Running Time} \\
        \hline
        Bellman~\cite{Bellman57} & $O(n \cdot \min\{W, \OPT\})$  \\
        Pisinger~\cite{Pisinger99} & $O(n \cdot \pmax \cdot \wmax)$   \\
        Kellerer and Pferschy~\cite{KellererP04}, also \cite{BateniHSS18,AxiotisT19} & $\Ohtilde(n + \wmax \cdot W)$   \\
        Bateni, Hajiaghayi, Seddighin and Stein~\cite{BateniHSS18} & $\Ohtilde(n + \pmax \cdot W)$  \\
        Axiotis and Tzamos~\cite{AxiotisT19} & $\Ohtilde(n \cdot \min\set{\wmax^2, \pmax^2})$ \\
        Bateni, Hajiaghayi, Seddighin and Stein~\cite{BateniHSS18} & $\Ohtilde((n + W) \cdot \min\set{\wmax, \pmax})$ \\
        Polak, Rohwedder and Węgrzycki~\cite{PolakRW21} & $O(n + \min\set{\wmax^3, \pmax^3})$  \\
        Bringmann and Cassis~\cite{BringmannC22} & $\Ohtilde(n + (W + \OPT)^{1.5})$  \\
        \textbf{\Cref{thm:alg-knapsack}}  & $\Ohtilde(n \cdot \wmax \cdot \pmax^{2/3})$  \\
        \textbf{\Cref{thm:alg-knapsack-symmetric}} & $\Ohtilde(n \cdot \pmax \cdot \wmax^{2/3})$  \\
        \hline
    \end{tabular}
\end{table}
\egroup

\paragraph{Min-Plus Convolution}

Given functions $f, g \colon \rangezero{n} \mapsto \Int$, their min-plus convolution is the function $h \colon \rangezero{2n} \mapsto \Int$ defined as $h(x) = \min_{x'} f(x') + g(x - x')$ for $x \in \rangezero{2n}$.
This can be trivially computed in time $O(n^2)$, and the best known algorithm for it runs in time $n^2/2^{\Omega(\sqrt{\log n})}$~\cite{BremnerCDEHILPT14,Williams18,ChanW21}.
The lack of faster algorithms has led to the \emph{Min-Plus Convolution Hypothesis}, which postulates that there is no truly subquadratic algorithm for this problem~\cite{CyganMWW19,KunnemannPS17}.
Despite this hypothesis, there are structured instances of min-plus convolution that can be solved faster~\cite{AggarwalKMSW87,BateniHSS18,BussieckHWZ94,ChanL15,ChiDX022}.
These improvements have been \emph{key to obtain the Knapsack algorithms listed in~\cref{tab:running-times}} (the only exception being Bellman's and Pisinger's algorithms~\cite{Bellman57,Pisinger99}):
\begin{itemize}
    \item When one of the functions is convex, their min-plus convolution can be computed in time $O(n)$ using the SMAWK algorithm~\cite{AggarwalKMSW87}.
    This has been used for Knapsack indirectly\footnote{Kellerer and Pferschy did not use SMAWK, but gave a different algorithm for computing the min-plus convolution of these instances in time $O(n\log n)$.} by Kellerer and Pferschy~\cite{KellererP04}, and explicitly by Axiotis and Tzamos~\cite{AxiotisT19} and Polak et al.~\cite{PolakRW21}.
    \item When the functions are monotone and have bounded entries, their min-plus convolution can be computed in time $\Ohtilde(n^{1.5})$ by an algorithm due to Chi et al.~\cite{ChiDX022}.
    This has been used for Knapsack by Bringmann and Cassis~\cite{BringmannC22}.
    \item Bateni et al.~\cite{BateniHSS18} introduced the \emph{prediction technique} to show that the min-plus convolution of certain instances arising from Knapsack can be computed efficiently.
    More precisely, let $h$ be the min-plus convolution of two given functions $f, g \colon \rangezero{n} \mapsto \Int$.
    They show that if one is given $n$ intervals $\range{x_i}{y_i}$ for $i \in \rangezero{n}$ satisfying
    (i) $|h(i+j) - (f(i) + g(j))| \leq \Delta$ for every $i \in \rangezero{n}$ and $j \in \range{x_i}{y_i}$,
    (ii) for every output $h(k)$ there exists at least one $i$ such that $f(i) + g(k-i) = h(k)$ and $k-i \in \range{x_i}{y_i}$
    and (iii) $0 \leq x_i, y_i <  n$ for all intervals and $x_i \leq x_j, y_i \leq y_j$ for all $i < j$; then $h$ can be computed in time $\Ohtilde(n\cdot \Delta)$.
    They showed that this is applicable in the context of Knapsack.
\end{itemize}

Our \cref{thm:alg-knapsack,thm:alg-knapsack-symmetric} fall into the same category of improvements, as we design an efficient algorithm for a new class of structured instances of min-plus convolution, namely \emph{near convex} functions:
We say that $f \colon \rangezero{n} \mapsto \Int$ is $\Delta$-near convex, if there is a convex function $\cvx{f} \colon \rangezero{n} \mapsto \Rat$ such that $\cvx{f}(i) \leq f(i) \leq \cvx{f}(i) + \Delta$ for all $i \in \rangezero{n}$.
Our theorem reads as follows:

\begin{restatable}[Near Convex MinPlus Convolution]{theorem}{thmconvexminplus} \label{thm:convex-minplus}
    Let $f \colon \rangezero{n} \mapsto \range{-U}{U}$, and $g \colon \rangezero{m} \mapsto \range{-U}{U}$ be given as inputs where $n,m,U \in \Nat$.
    Let $\Delta \geq 1$ such that both $f$ and $g$ are $\Delta$-near convex.
    Then the min-plus convolution of $f$ and $g$ can be computed in time $\Ohtilde((n+m) \cdot \Delta)$.
\end{restatable}

We view our~\cref{thm:convex-minplus} as a replacement for the prediction technique by Bateni et al.~\cite{BateniHSS18}. Indeed, all uses of the prediction technique exploit near-convexity to ensure its preconditions, and thus all uses that we are aware of can be replaced by our~\cref{thm:convex-minplus}. Since the prediction technique is both difficult to state and difficult to apply, we view our~\cref{thm:convex-minplus} as replacing the prediction technique by an \emph{easily applicable tool with a concise statement}.
Moreover, \cref{thm:convex-minplus} provides a new tool for structured instances of min-plus convolution, which we use in this paper to make progress on 0-1-Knapsack, and which we believe has wider applicability.

\paragraph{Our Techniques}
Our approach to prove~\cref{thm:convex-minplus} is as follows.
Let $f, g \colon \rangezero{n} \mapsto \Int$ be the input functions, and let $h$ be their min-plus convolution, which we aim to compute.
First we observe that we can obtain the convex approximations $\cvx{f}, \cvx{g}$ witnessing the $\Delta$-near convexity of $f$ and $g$,
and compute their min-plus convolution $\cvx{h}$ efficiently.
By exploiting $\cvx{h}$ and the convexity of $\cvx{f}$ and  $\cvx{g}$, we identify a structured set $R \subseteq \rangezero{n}^2$
with the property that any $(i, j) \in \rangezero{n}^2 \setminus R$ satisfies $f(i) + g(j) > h(j)$.
Then, we give a simple recursive algorithm to cover $R$ with a collection $\mathcal{C}$ of disjoint dyadic boxes $I \times J$ , where $(I, J) \in \mathcal{C}$ satisfies  $I, J \subseteq \rangezero{n}$ and $I \times J \subseteq R$.
Thus, we can infer $h$ by computing the \emph{sumset} $A := \set{(i, f(i)) \mid i \in I} + \set{(j, g(j)) \mid j \in J}$ and taking $h(k) = \min\set{y \mid (k, y) \in A}$ for every $(I, J) \in \mathcal{C}$.
To do this efficiently we observe that inside $I$ and $J$, the functions $f[I]$ and $g[J]$ are close to linear functions with the same slope up to an additive error of $\pm O(\Delta)$ (which follows from their $\Delta$-near convexity).
This implies that their sumset is small; more precisely it has size $O((|I|+|J|)\Delta)$.
Finally, we make use of known tools that can compute a sumset in time proportional to its size.
The idea of identifying a covering with small sumsets to efficiently compute the min-plus convolution is inspired by Chan and Lewenstein's~\cite{ChanL15}  algorithm for bounded monotone sequences (in which they do not use convexity in any form).
Our algorithm shares some similarities with the prediction technique by Bateni et al.~\cite{BateniHSS18}.
In particular, the covering by dyadic boxes where functions are near-linear resembles the way in which they exploit the intervals $\range{x_i}{y_i}$ required by their algorithm.

To obtain~\cref{thm:alg-knapsack,thm:alg-knapsack-symmetric}, we follow the \emph{partition and convolve} paradigm that has been used in many recent algorithms for Subset Sum and Knapsack, see e.g.~\cite{Bringmann17,BateniHSS18,BringmannC22,JansenR19,Chan18a,KoiliarisX19,DengJM23,BringmannN21}.
Specifically, we randomly split the items into $q$ groups. In each group, we use the standard dynamic programming algorithm to compute for each weight $i$ the maximum profit $f(i)$ attainable with weight at most $i$ using items from that group. Then we combine the functions $f$ over all groups by min-plus convolution. The crucial observation is that due to the random splitting we only need to compute the values $f(i)$ in a small weight interval. %

\paragraph{Further Related Work} \label{sec:related-work}

Cygan et al.~\cite{CyganMWW19} and K\"{u}nnemann et al.~\cite{KunnemannPS17} showed that under the Min-Plus Convolution Hypothesis,
there is no truly subquadratic algorithm for Knapsack on instances with $\wmax, W = \Theta(n)$ and $\pmax,\OPT = \Omega(n^2)$,
and symmetrically, on instances with $\pmax, \OPT = \Theta(n)$ and $\wmax, W = \Omega(n^2)$.
This implies that Bellman's dynamic programming algorithm is conditionally optimal in these settings.

Pseudopolynomial-time algorithms parameterized by $\pmax$ and $\wmax$ have also been studied for the closely related Unbounded Knapsack problem.
Here, the setup is the same as for 0-1 Knapsack but now a solution might include an arbitrary number of copies of each item.
Chan and He~\cite{ChanH22} gave an algorithm for this problem in time $\Ohtilde(n \cdot \min\set{\pmax, \wmax})$, which is optimal under the Min-Plus Convolution Hypothesis.
Bringmann and Cassis~\cite{BringmannC22} gave an algorithm in time $\Ohtilde(n + (\pmax + \wmax)^{1.5})$ which is better when $\wmax \approx \pmax \approx n$.

\paragraph{Outline} The paper is organized as follows. In~\cref{sec:preliminaries} we give some formal preliminaries and establish some notation. In~\cref{sec:knapsack-algorithm} we give our algorithm for Knapsack proving~\cref{thm:alg-knapsack,thm:alg-knapsack-symmetric}, assuming~\cref{thm:convex-minplus}. In~\cref{sec:minplus-algorithm} we will then give our algorithm for min-plus convolution, proving~\cref{thm:convex-minplus}.

%% file: sections/preliminaries.tex
\section{Preliminaries}\label{sec:preliminaries}

We write $\Nat = \set{0,1,2,\dots}$. For $t \in \Nat$, we define $\rangezero{t} := \set{0,1,\dots,t}$.
Let $A \in \Int^{n+1}$ be an integer sequence, i.e., $A[i] \in \Int$ for $i \in \rangezero{n}$.
Sometimes we will refer to such a sequence as a \emph{function} $A \colon \rangezero{n} \mapsto \Int$.
With this in mind, we use the notation $-A$ to denote the entry-wise negation of $A$.
Given $a, b \in \Real$ with $a \leq b$, we define $\range{a}{b} := \set{\max(0, \floor{a}), \max(0, \floor{a})+1,\dots,\ceil{b}-1,\ceil{b}}$.
The non-standard rounding and capping at 0 in the definition of $\range{a}{b}$
is useful to index a subsequence $A\range{a}{b}$ when $a$ and $b$ might not be non-negative integers.

The max-plus convolution of two sequences $A\range{0}{n} \in \Int^{n+1}, B\range{0}{m} \in \Int^{m+1}$, denoted by $\maxconv{A}{B}$,
is a sequence of length $n+m+1$ where for each $k \in \rangezero{n+m}$ we have $\maxconv{A}{B}[k] := \max_{i + j = k}A[i] + B[j]$.
The min-plus convolution $\minconv{A}{B}$ is defined analogously, but replacing max by a min.
Note that by negating the entries of the sequences, these two operations are equivalent.

\begin{fact}\label{fact:equiv-convs}
    For any $A \in \Int^{n+1}, B \in \Int^{m+1}$, we have $\maxconv{A}{B} = -\minconv{-A}{-B}$.
\end{fact}

We will use the following handy notation: Given sequences $A\range{0}{n}, B\range{0}{n}$ and
intervals $I, J \subseteq [n]$ and $K \subseteq [2n]$, we denote by $C[K] := \maxconv{A[I]}{B[J]}$ the computation of
$C[k] := \max\{A[i] + B[j] \colon i \in I, j \in J, i + j = k\}$ for each $k \in K$.

We say that a function $f \colon \rangezero{n} \mapsto \Rat$ is \emph{convex} if $f(i) - f(i-1) \leq f(i+1) - f(i)$ holds for every $i \in \range{1}{n-1}$.
We say that $f$ is \emph{concave} if $-f$ is convex.

\begin{definition}[Near Convex and Near Concave Functions]\label{def:nearconvex}
    For $\Delta \geq 0$, we say that a function $f \colon\rangezero{n} \mapsto \Int$ is $\Delta$-near convex,
    if there is a convex function $\cvx{f} \colon \rangezero{n} \mapsto \Rat$ such that $\cvx{f}(i) \leq f(i) \leq \cvx{f}(i) + \Delta$.
    We say that $f$ is $\Delta$-near concave if $-f$ is $\Delta$-near convex.
\end{definition}

If the input consists of $N$ numbers in $\range{-U}{U}$, we denote $\Ohtilde(T) = \bigcup_{c \geq 0} O(T \log^c(NU))$.

%% file: sections/knapsack-algorithm.tex
\section{Faster 0-1 Knapsack Algorithm}\label{sec:knapsack-algorithm}
In this section we prove~\cref{thm:alg-knapsack}.
Let $(\calI, W)$ be a 0-1 Knapsack instance. Throughout, we denote the number of items by $n := |\calI|$.
We identify the item set $\calI$ with $\set{1,\dots,n}$.
We represent a \emph{solution} to the knapsack instance (i.e., a subset of $\calI$), by an indicator vector $x \in \set{0,1}^n$.
For a subset of the items $\calJ \subseteq \calI$, we put $w_\calJ(x) := \sum_{i \in \calJ}w_i x_i$ and $p_\calJ(x) := \sum_{i \in \calJ} p_i x_i$.
We define the profit sequence $\Profit{\calI}[\cdot]$, where for each $j \in \Nat$ we have
\[
	\Profit{\calI}[j] = \max\set{p_{\calI}(x) \mid x \in \set{0,1}^n, w_{\calI}(x) \leq j}.
\]
Observe that $\Profit{\calI}$ is monotone non-decreasing, and that $\OPT = \Profit{\calI}[W]$.
The textbook way to compute $\Profit{\calI}\range{0}{j}$ is to use dynamic programming:

\begin{fact}\label{fact:dp}
    For any $j \in \Nat$ the sequence $\Profit{\calI}\range{0}{j}$ can be computed in time $O(nj)$.
\end{fact}

Before presenting the algorithm, we make two simple observations about the given Knapsack instance $(\calI, W)$.
First, by ignoring items with weight larger than the capacity $W$, we can assume without loss of generality that $\wmax \leq W$.
Now every single item is a feasible solution, so we have $\pmax \leq \OPT$.
Second, observe that if $W \geq n \cdot \wmax$, then the instance is trivial since we can pack all items.
Thus, we can assume without loss of generality that $W \leq n \cdot \wmax$.
Moreover, since any feasible solution consists of at most all the $n$ items, it follows that $\OPT \leq n \cdot \pmax$.

\subsection*{The Algorithm}

We now describe the algorithm.
Set parameters $q :=\min\set{(n/\pmax)^{2/3}(W/\wmax)^{1/3}, W/\wmax}$ rounded down to the closest power of 2, $\Delta := \wmax W /q$ and $\eta := 11 \log n$.
For each $\ell \in \rangezero{\log q}$ we define the interval $J^\ell := \range{\tfrac{W}{q}2^\ell - \sqrt{\Delta 2^\ell}\eta}{\tfrac{W}{q}2^\ell + \sqrt{\Delta 2^\ell}\eta}$.

We start by splitting the items $\calI$ into $q$ groups $\calI^0_1,\dots,\calI^0_q$ uniformly at random.
The idea will be to compute an array $C_j^0$ associated to each $\calI_j^0$, and then combine them in a \emph{tree-like} fashion.
A crucial aspect for the running time is that we only compute $|J^\ell|$ entries of each array $C_j^\ell$.
In detail, we proceed as follows:

\begin{description}
    \item[Base Case] For each $\calI^0_j$, we use~\cref{fact:dp} to compute
    $\Profit{\calI^0_j}\range{0}{\tfrac{W}{q} + \sqrt{\Delta}\eta}$ and define the subarray
    $C_j^0[J^0] := \Profit{\calI^0_j}[J^0]$.
    \item[Combination] Iterate over the \emph{levels} $\ell = 1,\dots q$.
    For $j \in \range{1}{q/2^\ell}$ we set $\calI^\ell_j := \calI^{\ell-1}_{2j-1} \cup \calI^{\ell-1}_{2j}$.
    Then, compute the subarray $C_j^{\ell}[J^\ell]$ by taking the relevant entries of the max-plus convolution of $C_{2j-1}^{\ell-1}[J^{\ell-1}]$ and $C_{2j}^{\ell-1}[J^{\ell-1}]$.
    \item[Returning the answer] (Note that when $\ell = \log(q)$, it holds that $\calI^{\log q}_1 = \calI$.)
    We return the value $C_1^{\log q}[W]$.
    See~\cref{alg:knapsack} for the pseudocode.
\end{description}

\begin{algorithm}[ht]
    \caption{Knapsack Algorithm. Given a set of items $\calI$ and a weight budget $W$, the algorithm computes the maximum attainable profit.}\label{alg:knapsack}
    \begin{algorithmic}[1]
    \State $q \gets \min\set{(n/\pmax)^{2/3}(W/\wmax)^{1/3}, W/\wmax}$ rounded down to the closest power of 2
    \State $\Delta \gets \wmax W / q$
    \State $\eta \gets 11 \log n$
    \State $\calI_1^0, \dots, \calI_q^0 \gets$ random partitioning of $\calI$ into $q$ groups
    \smallskip
    \For{$i = 1\dots q$}\label{alg:knapsack:line:for-loop-basecase}
        \State Compute $\Profit{\calI^0_j}\range{0}{\tfrac{W}{q} + \sqrt{\Delta}\eta}$ using standard dynamic programming (\cref{fact:dp}) \label{alg:knapsack:line:base-dp}
        \State $J^0 \gets \range{\tfrac{W}{q} - \sqrt{\Delta}\eta}{\tfrac{W}{q} + \sqrt{\Delta}\eta}$ \label{alg:knapsack:line:range-base}
        \State $C_j^0[J^0] \gets \Profit{\calI^0_j}[J^0]$ \label{alg:knapsack:line:set-base}
    \EndFor
    \smallskip
    \For{$\ell = 1\dots\log(q)$}\label{alg:knapsack:line:for-combining}
        \State $J^\ell \gets \range{\tfrac{W}{q}2^\ell - \sqrt{\Delta 2^\ell}\eta}{\tfrac{W}{q}2^\ell + \sqrt{\Delta 2^\ell}\eta}$
        \For{$j = 1,\dots,q/2^\ell$}
            \State $\calI_j^\ell \gets \calI_{2j-1}^{\ell-1} \cup \calI_{2j}^{\ell-1}$
            \State Compute $C^\ell_j[J^\ell] \gets \maxconv{C^{\ell-1}_{2j-1}[J^{\ell-1}]}{C^{\ell-1}_{2j}[J^{\ell-1}]}$ using~\cref{thm:convex-minplus} \label{alg:knapsack:line:prediction}
        \EndFor
    \EndFor
    \smallskip
    \State \Return $C^{\log q}_1 [W]$\label{alg:knapsack:line:return}
\end{algorithmic}
\end{algorithm}

\subsection*{Correctness}
We start by analyzing the correctness of the algorithm.
The following lemma shows that the weight of any solution restricted to one of the sets $\calI^\ell_j$ is concentrated around its expectation.

\begin{restatable}[Concentration]{lemma}{lemconcentration}\label{lem:concentration}
    Let $x \in \set{0,1}^n$ be a solution to the given Knapsack instance.
    Fix a level $\ell \in \range{0}{\log q}$ and $j \in \range{1}{q/2^\ell}$.
    Then, with probability at least $1 - 1/n^4$ it holds that:
    \[
        \left|w_{\calI^\ell_j}(x) - \frac{w_{\calI}(x) \cdot 2^\ell}{q}\right| \leq \sqrt{\Delta 2^\ell} \cdot 10 \log n.
    \]
\end{restatable}
\begin{proof}
    Recall that the item set $\calI$ is partitioned randomly into $\calI^0_1,\dots,\calI^0_q$.
    Thus, observe that $\calI^\ell_j$ is a random subset of $\calI$, where each item is included with probability $p := 2^\ell/q$.
    For $i \in \range{1}{n}$, let $Z_i$ be a random variable which equals $w_i \cdot x_i$ with probability $p$, and $0$ with probability $1 - p$.
    Then, observe that $w_{\calI^\ell_j}(x)$ is distributed as $Z := \sum_{i=1}^n Z_i$, and therefore, $\Ex(Z) = w_{\calI}(x) p$.

    To prove the statement, we will use Bernstein's inequality (see e.g. \cite[Theorem 1.2]{DubhashiP09}) which states that
    \begin{align}
        \Pr(|Z - \Ex(Z)| \geq t)
            &\leq 2\exp\left(-\frac{t^2}{2\Var(Z) + \tfrac{2}{3} t \cdot \wmax}\right) \nonumber \\
            &\leq 2\exp\left(-\min\left\{\frac{t^2}{4 \Var(Z)}, \frac{t}{2\wmax}\right\}\right). \label{proof:concentration:eqn:1}
    \end{align}
    Set $t := \sqrt{p \cdot \wmax W} \cdot 10 \log n$.
    We first bound $t^2/(4 \Var(Z))$.
    Note that we can give an upper bound on the variance as follows:
    \[
        \Var(Z) = \sum_{i=1}^n p(1-p)w_i^2 x_i^2 \leq p \cdot \wmax \sum_{i=1}^n w_i x_i = p \cdot \wmax w_{\calI}(x) \leq p \cdot \wmax W.
    \]
    Therefore $t^2 / (4 \Var(Z)) \geq 10 \log n$.
    Next, we bound $t/(2\wmax)$. Using that $q \leq W/\wmax$, we have that $p = \tfrac{2^\ell}{q} \geq \tfrac{\wmax 2^\ell}{W} \geq \tfrac{\wmax}{W}$.
    Thus,
    \[
        \frac{t}{2\wmax} = \frac{\sqrt{p \cdot \wmax W} \cdot 10 \log n}{2\wmax} \geq 5 \log n.
    \]
    Combining the above, we obtain from~\eqref{proof:concentration:eqn:1} that
    \[
        |w_{\calI^\ell_j}(x) - w_\calI(x)2^\ell/q| = |Z - \Ex(Z)| \leq t = \sqrt{p \wmax W} \cdot 10\log n = \sqrt{\Delta 2^\ell} \cdot 10\log n
    \]
    holds with probability at least $1 - 2/n^5 \geq 1 - 1/n^{4}$.
\end{proof}

Using \cref{lem:concentration}, we can argue that at level $\ell$ it suffices to compute a subarray of length $\Ohtilde(\sqrt{\Delta 2^\ell})$ around $W2^\ell/q$.
The following lemma makes this precise:

\begin{restatable}{lemma}{lemcorrectness}\label{lem:correctness-1}
    Let $x \in \set{0,1}^n$ be a solution to the given Knapsack instance satisfying $w_\calI(x) \in \range{W - \wmax}{W}$.
    With probability at least $1 - 1/n^2$, for all levels $\ell \in \range{0}{\log q}$ and all $j \in \range{1}{q/2^\ell}$ it holds that:
    \begin{itemize}
        \item $w_{\calI^\ell_j}(x) \in J^\ell = \range{\tfrac{W}{q}2^\ell - \sqrt{\Delta 2^\ell}\eta}{\tfrac{W}{q}2^\ell + \sqrt{\Delta 2^\ell}\eta}$, and
        \item $C^\ell_j[w_{\calI^\ell_j}(x)] \geq p_{\calI^\ell_j}(x)$.
    \end{itemize}
\end{restatable}
\begin{proof}
    By~\cref{lem:concentration}, for each $\ell \in \range{0}{\log q}$ and $j \in \range{1}{q/2^\ell}$ it holds that
    \begin{equation}\label{proof:correctness:eqn:1}
        |w_{\calI_j^\ell}(x) - w_\calI(x)2^\ell/q| \leq \sqrt{\Delta 2^\ell} \cdot 10\log n
    \end{equation}
    with probability at least $1 - 1/n^4$.
    Note that $q \leq W/\wmax \leq n$.
    Thus, we can afford a union bound and conclude that \eqref{proof:correctness:eqn:1} holds \emph{for all} $\ell \in \range{0}{\log q}$ and $j \in \range{1}{q/2^\ell}$ with probability at least $1 - 1/n^2$.
    From now on, we condition on this event.

    We start by showing the first item of the statement.
    Fix $\ell \in \rangezero{\log q}$ and $j \in \range{1}{q/2^\ell}$.
    By \eqref{proof:correctness:eqn:1}, it holds that $|w_{\calI_j^\ell}(x) - w_\calI(x)2^\ell/q| \leq \sqrt{\Delta 2^\ell} \cdot 10 \log n$.
    By assumption, we have that $w_\calI(x) \in \range{W - \wmax}{W}$. Hence,
    \begin{align*}
        |w_{\calI^\ell_j}(x) - W2^\ell/q|
            &\leq  |w_{\calI^\ell_j}(x) - w_\calI(x)2^\ell/q| + \tfrac{2^\ell}{q}|w_\calI(x) - W| \\
            &\leq \sqrt{\Delta 2^\ell} 10\log n + \wmax 2^\ell/q \leq \sqrt{\Delta 2^\ell} \cdot 11\log n.
    \end{align*}
    The last inequality holds since we can use that $2^\ell \leq q$ and $\wmax \leq W$  to obtain that $\wmax 2^\ell/q \leq \sqrt{\wmax W} \cdot \sqrt{2^\ell/q} = \sqrt{\Delta 2^\ell}$.
    Since $\eta = 11\log n$, this implies that $w_{\calI^\ell_j}(x) \in J^\ell = \range{\tfrac{W}{q}2^\ell - \sqrt{\Delta2^\ell}\eta}{\tfrac{W}{q}2^\ell + \sqrt{\Delta2^\ell}\eta}$. This concludes the proof of the first item.

    Next, we prove the second item of the lemma by induction. Consider the base case $\ell = 0$.
    By the first item, for any $j \in \range{1}{q}$ we have that $w_{\calI^0_j}(x) \in J^0$.
    In particular, it holds that $C_j^0[w_{\calI^0_j}(x)] = \Profit{\calI^0_j}[w_{\calI^0_j}(x)]$ (see~\cref{alg:knapsack:line:set-base}).
    Then, since $\Profit{\calI^0_j}[i]$ is the maximum profit of a subset of items from $\calI^0_j$ of weight at most $i$, it holds that $\Profit{\calI^0_j}[w_{\calI^0_j}(x)] \geq  p_{\calI^0_j}(x)$, which completes the proof of the base case.

    Now we proceed with the inductive step: Fix $\ell \ge 1$ and assume that
    $C_j^{\ell-1}[w_{\calI_j^{\ell-1}}(x)] \geq p_{\calI_j^{\ell-1}}(x)$
    hold for all $j \in \range{1}{q/2^{\ell-1}}$.
    By the first item of the lemma, for each $j \in \range{1}{q/2^\ell}$ we have that $w_{\calI_j^\ell}(x) \in J^\ell$.
    Thus, by the computation of~\cref{alg:knapsack:line:prediction}, it holds that
    \begin{align*}
        C_j^\ell[w_{\calI_j^\ell}(x)]
            &= \max \set{C_{2j-1}^{\ell-1}[i] + C_{2j}^{\ell}[i'] \colon i,i' \in J^{\ell-1},\, i+i' = w_{\calI_j^\ell}(x)} \\
            &\geq C_{2j-1}^{\ell-1}[w_{\calI_{2j-1}^{\ell-1}}(x)] + C_{2j}^{\ell}[w_{\calI_{2j}^{\ell-1}}(x)] \\
            &\geq p_{\calI_{2j-1}^{\ell-1}}(x) + p_{\calI_{2j}^{\ell-1}}(x)
            = p_{\calI_j^\ell}(x).
    \end{align*}
    In the second step, we used that $w_{\calI_{2j-1}^{\ell-1}}(x), w_{\calI_{2j}^{\ell-1}}(x) \in J^{\ell-1}$ as shown earlier.
    The third step follows from the induction hypothesis.
    The last equality holds since $\calI_j^\ell = \calI_{2j-1}^{\ell-1} \cup \calI_{2j}^{\ell-1}$.
\end{proof}

\begin{lemma}[Correctness of~\cref{alg:knapsack}]\label{lem:correctness-2}
    Let $x^* \in \set{0,1}^n$ be an optimal solution to the given Knapsack instance.
    Then, for every $i \in \range{w_{\calI}(x^*)}{W}$, it holds that $C_1^{\log q}[i] = \Profit{\calI}[i]$ with probability at least $1 - 1/n^2$.
\end{lemma}
\begin{proof}
    We can check in linear time $O(n)$ whether the optimal solution consists of all items, in which case the instance is trivial.
    Thus, we can assume without loss of generality that $x^*$ does not include all items.
    In particular, $x^*$ leaves at least one item out and therefore its weight satisfies $w_\calI(x^*) \in \range{W - \wmax}{W}$.
    By~\cref{lem:correctness-1}, it holds that $C_1^{\log q}[w_\calI(x^*)] \geq p_\calI(x^*) = \Profit{\calI}[w_\calI(x^*)]$ with probability at least $1 - 1/n^2$.
    From now on we condition on this event.
    We will use the following auxiliary claim:

    \begin{claim}\label{proof:correctness-1:claim:1}
        The sequence $C_1^{\log q}[J^{\log q}]$ is monotone non-decreasing, and satisfies $C_1^{\log q}[i] \leq \Profit{\calI}[i]$ for all $i \in J^{\log q}$.
    \end{claim}
    \begin{claimproof}
        First we argue monotonicity by induction.
        Note that in the base case $\ell = 0$, the sequence $C_j^0[J^0] = \Profit{\calI_j^0}[J^0]$ is monotone non-decreasing due to the definition of $\Profit{\calI_j^0}$.
        For level $\ell > 0$, the sequence $C_j^\ell$ is computed by taking the max-plus convolution of sequences of level $\ell - 1$.
        The result follows by observing that the max-plus convolution of two monotone non-decreasing sequences is monotone non-decreasing.

        The second part of the claim follows since (inductively) every entry $C_1^{\log q}[i]$ for $i \in J^{\log q}$ corresponds to the profit of a subset of items of $\calI$ of weight at most $i$.
    \end{claimproof}

    Since $x^*$ is an optimal solution, it holds that $\Profit{\calI}[i] = p_\calI(x^*)$ for all $i \in \range{w_\calI(x^*)}{W}$.
    Thus \cref{proof:correctness-1:claim:1} yields that $C_1^{\log q}[i] = \Profit{\calI}[i]$ for all $i \in \range{w_\calI(x^*)}{W}$, completing the proof.
\end{proof}

\subsection*{Running Time}

Now we analyze the running time of~\cref{alg:knapsack}.
The key speedup comes from the computation in~\cref{alg:knapsack:line:prediction}, where we use~\cref{thm:convex-minplus} to perform the max-plus convolution.
Since~\cref{thm:convex-minplus} is phrased in terms of min-plus convolution of near-convex functions, we will use the following corollary:

\begin{corollary}\label{cor:concave-maxplus}
    Let $f \colon \rangezero{n} \mapsto \range{-U}{U}$ and $g \colon \rangezero{m} \mapsto \range{-U}{U}$ be given as inputs, where $U \in \Nat$.
    Let $\Delta \geq 1$ such that both $f$ and $g$ are $\Delta$-near concave.
    Then, $\maxconv{f}{g}$ can be computed in time $\Ohtilde((n+m)\Delta)$
\end{corollary}
\begin{proof}
    Noting that $-f$ and $-g$ are $\Delta$-near convex (\cref{def:nearconvex}), the result follows from \cref{thm:convex-minplus,fact:equiv-convs}.
\end{proof}

The following lemma shows that the max-plus convolution of two near-concave sequences remains near-concave.

\newcommand{\ihat}{\hat\imath}
\begin{lemma}\label{lem:minconv-preserves-nearcvx}
	Let $f \colon\rangezero{n} \mapsto \Int$ be $\Delta_f$-near concave and  $g \colon\rangezero{m} \mapsto \Int$ be $\Delta_g$-near concave.
    Then, $h := \maxconv{f}{g}$ is $\Delta_h$-near concave with $\Delta_h \leq \max\set{\Delta_f, \Delta_g}$.
\end{lemma}
\begin{proof}
    Let $\cvx{f}, \cvx{g}$ be pointwise minimal concave functions with $\cvx{f} \ge f$, $\cvx{g} \ge g$
    and let $\cvx{h} := \maxconv{\cvx{f}}{\cvx{g}}$.
    We will show that $\cvx{h} \ge h \ge \cvx{h} - \max\set{\Delta_f, \Delta_g}$, which implies the statement.

    To show that $\cvx{h} \ge h$, fix $k \in \rangezero{n+m}$ and let $i^*$ be a witness for $h(k)$, i.e., $h(k) = f(i^*) + g(k-i^*)$.
    Then, $\cvx{h}(k) \ge \cvx{f}(i^*) + \cvx{g}(k-i^*) \ge f(i^*) + g(k-i^*) = h(k)$. So $\cvx{h} \ge h$.

    To show that $h \ge \cvx{h} - \Delta$ for $\Delta := \max\set{\Delta_f, \Delta_g}$, fix $k \in \rangezero{n+m}$ and let $i^*$ be a witness for $\cvx{h}(k)$, i.e., $\cvx{h}(k) = \cvx{f}(i^*) + \cvx{g}(k-i^*)$.
    Note that $\cvx{f}$ is piecewise a linear interpolation between points on $f$.
    In particular, there exist $i_L \leq i^* \leq i_R$ such that $\cvx{f}(i_L) = f(i_L), \cvx{f}(i_R) = f(i_R)$ and
    $\cvx{f}(i)$ is linear for $i \in \range{i_L}{i_R}$.
    Similarly, for $j^* := k - i^*$ there exist $j_L \leq j^* \leq j_R$ such that $\cvx{g}(j_L) = g(j_L), \cvx{g}(j_R) = g(j_R)$
    and $\cvx{g}(j)$ is linear for $j \in \range{j_L}{j_R}$.
    We pick the maximum $i_L, j_L$ and minimum $i_R, j_R$ with this property.

    Let $\ihat_L := \max\set{i_L, k - j_R}, \ihat_R := \min\set{i_R, k-j_L}$.
    Observe that the function $\cvx{s}(i) := \cvx{f}(i) + \cvx{g}(k-i)$ is linear for $i \in \range{\ihat_L}{\ihat_R}$,
    and that $\ihat_L \leq i^* \leq \ihat_R$.
    Moreover, by definition of $\cvx{h}$ we have that $\cvx{s}(i) = \cvx{f}(i) + \cvx{g}(k-i) \le \cvx{h}(k)$ for $i \in \range{\ihat_L}{\ihat_R}$.
    Since $i^*$ is a witness of $\cvx{h}(k)$, we have $\cvx{s}(i^*) = \cvx{h}(k)$.
    Combining the above, we obtain that $\cvx{s}(i) = \cvx{h}(k)$ for all $i \in \range{\ihat_L}{\ihat_R}$.
    In particular, $\cvx{f}(\ihat_L) + \cvx{g}(k-\ihat_L) = \cvx{h}(k)$, and thus
    \begin{align}
        h(k)
            &\ge f(\ihat_L) + g(k-\ihat_L) \nonumber \\
            &= \cvx{f}(\ihat_L) + \cvx{g}(k - \ihat_L) + (f(\ihat_L) - \cvx{f}(\ihat_L)) + (g(k-\ihat_L) - \cvx{g}(k - \ihat_L)) \nonumber \\
            &= \cvx{h}(k) + (f(\ihat_L) - \cvx{f}(\ihat_L)) + (g(k-\ihat_L) - \cvx{g}(k - \ihat_L)). \label{proof:minconv-preserves-nearcvx:eqn:1}
    \end{align}
    Finally, since $\ihat_L \in \set{i_L, k - j_R}$ and $f(i_L) = \cvx{f}(i_L)$ and $g(j_R) = \cvx{g}(j_R)$,
    one of the two last summands in~\eqref{proof:minconv-preserves-nearcvx:eqn:1} must be 0.
    Using the near-concavity of $f$ and $g$, we can bound the other summand by $f(\ihat_L) - \cvx{f}(\ihat_L) \ge -\Delta_f$ or
    $g(k-\ihat_L) - \cvx{g}(k-\ihat_L) \ge -\Delta_g$.
    This yields $h(k) \ge \cvx{h}(k) - \Delta_f$ or $h(k) \ge \cvx{h}(k) - \Delta_g$.
    In any case, we conclude that $h(k) \ge \cvx{h}(k) - \max\set{\Delta_f, \Delta_g}$ holds for every $k \in \rangezero{n+m}$.
\end{proof}

The next lemma shows that the sequences we combine in \cref{alg:knapsack:line:prediction} are near-concave.

\begin{restatable}[Near Concavity]{lemma}{lemnearconcavity}\label{lem:near-concavity}
    For every level $\ell \in \range{1}{q}$ and every $j \in \range{1}{q/2^\ell}$, it holds that $C^{\ell}_j[J^\ell]$ is $\pmax$-near concave.
\end{restatable}
\begin{proof}
    We prove the statement using induction.
    Focus in the base case $\ell = 0$.
    For each $j \in \range{1}{q}$, we have that $C_j^0[J^0] = \Profit{\calI_j^0}[J^0]$.
    In what follows, we argue that $\Profit{\calI_j^0}$ is $\pmax$-near concave.
    Consider the fractional greedy solution for Knapsack: sort the items $(p_1, w_1), \dots, (p_m, w_m)$ in $\calI_j^0$ non-decreasingly by their profit-to-weight ratio, i.e., so that $p_1/w_1 \geq p_2/w_2 \geq \cdots \geq p_m/w_m$.
    Let $M := \sum_{i=1}^m w_i$.
    Then, construct the sequence $\tilde{\calP}\range{0}{M}$ by setting breakpoints
    \[
        \tilde\calP[0] = 0,\, \tilde{\calP}[w_1] = p_1, \, \tilde{\calP}[w_1+w_2] = p_1+p_2, \, \dots,\, \tilde{\calP}[w_1 + \cdots + w_m] = p_1 + \cdots + p_m,
    \]
    and a linear interpolation between every pair of consecutive breakpoints.
    In this way, $\tilde\calP[i]$ corresponds the optimal solution to the \emph{fractional} version of Knapsack with capacity $i$, i.e.,
    in the setting where items can be fractionally packed in a solution.

    \begin{claim}\label{proof:combine-conv:claim:1}
        The sequence $\tilde\calP$ is concave, and it holds that $\tilde\calP[i] \geq \Profit{\calI_j^0}[i] \geq \tilde\calP[i] - \pmax$ for every $i \in \rangezero{M}$.
    \end{claim}
    \begin{claimproof}
        For each $i \in \range{1}{M-1}$ it holds that $\tilde\calP[i] - \tilde\calP[i-1] \geq \tilde\calP[i+1] - \tilde\calP[i]$ since the slopes of the linear pieces between breakpoints are non-decreasing due to the sorting by profit-to-weight ratio. This means that $\tilde\calP$ is concave.

        For each $i \in \rangezero{M}$, it holds that $\tilde\calP[i] \geq \Profit{\calI_j^0}[i]$ since $\tilde\calP[i]$ is the optimal solution of the fractional Knapsack.
        Moreover, observe that the solution attaining $\tilde\calP[i]$ contains at most one item allocated fractionally.
        By removing that item, we obtain a feasible (integral) solution to the 0-1 Knapsack of capacity $i$, and the profit is reduced by at most $\pmax$.
        This implies that $\Profit{\calI_j^0}[i] \geq \tilde\calP[i] - \pmax$.
    \end{claimproof}
    By \cref{proof:combine-conv:claim:1}, we conclude that $\Profit{\calI_j^0}\range{0}{M}$ is $\pmax$-near concave (see~\cref{def:nearconvex}), and therefore $C_j^0[J^0]$ is as well, which completes the proof of the base case.

    For the inductive step, consider a level $\ell > 0$.
    Fix a $j \in \range{1}{q/2^\ell}$.
    By the inductive hypothesis, $C_{2j-1}^{\ell-1}[J^{\ell-1}]$ and $C_{2j}^{\ell-1}[J^{\ell-1}]$ are $\pmax$-near concave.
    Thus, by~\cref{lem:minconv-preserves-nearcvx} we obtain that $C_j^\ell[J^\ell] = \maxconv{C^{\ell-1}_{2j-1}[J^{\ell-1}]}{C^{\ell-1}_{2j}[J^{\ell-1}]}$ is $\pmax$-near concave, completing the proof.
\end{proof}

\begin{lemma}\label{lem:combine-conv}
    Fix a level $\ell \in \range{1}{q}$ and an iteration $j \in \range{1}{q/2^\ell}$.
    The computation of $C_j^{\ell}$ in~\cref{alg:knapsack:line:prediction} takes time $\Ohtilde(\pmax\sqrt{\Delta 2^\ell})$
\end{lemma}
\begin{proof}
    By~\cref{lem:near-concavity}, the sequences $C_{2j-1}^{\ell-1}[J^{\ell-1}], C_{2j}^{\ell-1}[J^{\ell-1}]$ are $\pmax$-near concave.
    Thus, by \cref{cor:concave-maxplus}, their max-plus convolution can be computed in time $\Ohtilde(\pmax |J^{\ell}|) = \Ohtilde(\pmax\sqrt{\Delta 2^\ell})$, where we used $\eta = \Ohtilde(1)$.
\end{proof}

\begin{lemma}[Running Time of~\cref{alg:knapsack}]\label{lem:running-time}
    \cref{alg:knapsack} runs in time \[\Ohtilde((\pmax W)^{2/3} (n \wmax)^{1/3} + n\wmax).\]
\end{lemma}
\begin{proof}
    Recall that $q = \min\set{(n/\pmax)^{2/3}(W/\wmax)^{1/3}, W/\wmax}$ (up to a factor of 2).
    Since $W \leq n\wmax$, we have that $q \leq n$.
    Moreover, since we assume without loss of generality that $\wmax \leq n$, note that $q < 1$ if and only if
    $q = (n/\pmax)^{2/3}(W/\wmax)^{1/3} < 1$. This implies that $\pmax > n\sqrt{W/\wmax}$.
    But in this case, the claimed running time is $\Omega(nW)$, so the standard $O(nW)$ dynamic programming algorithm (\cref{fact:dp}) already achieves our time bound.
    Thus, we can assume without loss of generality that $1 \leq q \leq n$, i.e., $q$ is a valid choice for the number of groups in which we split the item set $\calI$.

    We start bounding the running time of the base case, i.e., the computation of the arrays $C_j^0$ for $j \in \range{1}{q}$ in~\cref{alg:knapsack:line:for-loop-basecase}.
    By~\cref{fact:dp}, and the definition $\Delta = \wmax W / q$ this takes time
    \begin{equation}\label{proof:running-time:eqn:1}
        O\left(\sum_{j=1}^q |\calI_j^0|(\tfrac{W}{q} + \sqrt{\Delta} \eta)\right)
            = O\left(n(\tfrac{W}{q} + \sqrt{\Delta}\eta)\right)
            = \Ohtilde\left(n\tfrac{W}{q} + n\sqrt{\tfrac{\wmax W}{q}}\right).
    \end{equation}

    Now we bound the time of the combination step done in~\crefrange{alg:knapsack:line:for-combining}{alg:knapsack:line:return}.
    At level $\ell \in \range{1}{q}$ and iteration $j \in \range{1}{q/2^\ell}$ the execution of~\cref{alg:knapsack:line:prediction}
    takes time $\Ohtilde(\pmax \sqrt{\Delta 2^\ell})$ by~\cref{lem:combine-conv}.
    Thus, we can bound the overall time as
    \begin{align*}
        \sum_{\ell=1}^{\log q}\sum_{j=1}^{q/2^{\ell}} \Ohtilde(\pmax \sqrt{\Delta 2^\ell})
            &= \sum_{\ell=1}^{\log q} \frac{q}{2^\ell}\Ohtilde\left(\pmax \sqrt{\tfrac{\wmax W}{q} 2^\ell} \,\right)
            = \sum_{\ell=1}^{\log q}\Ohtilde \left(\pmax\sqrt{\tfrac{q \wmax W}{2^\ell}}\right),
    \end{align*}
    since this is a geometric series, it is bounded by the first term $\Ohtilde(\pmax\sqrt{q \wmax W})$.
    Combining this with~\eqref{proof:running-time:eqn:1}, we obtain overall time
    \[
        \Ohtilde\left(\pmax \sqrt{q \wmax W} + n\tfrac{W}{q} + n \sqrt{\tfrac{\wmax W}{q}}\right).
    \]
    Recalling that $q = \Theta(\min\set{(n/\pmax)^{2/3}(W/\wmax)^{1/3}, W/\wmax})$, we obtain overall time
    \[
        \Ohtilde((\pmax W)^{2/3}(n\wmax)^{1/3} + n\wmax + (\pmax W)^{1/3}(n\wmax)^{2/3}).
    \]
    Finally, note that by the AM-GM inequality we have
    \begin{align*}
        (\pmax W)^{1/3}(n \wmax)^{2/3}
            &= \sqrt{(\pmax W)^{2/3} (n\wmax)^{1/3} n\wmax} \\
            &\leq O((\pmax W)^{2/3}(n\wmax)^{1/3} + n\wmax).
    \end{align*}
    Thus, the overall running time is $\Ohtilde((\pmax W)^{2/3}(n\wmax)^{1/3} + n\wmax)$, as claimed.
\end{proof}

\begin{proof}[Proof of~\cref{thm:alg-knapsack}]
    Run~\cref{alg:knapsack}.
    By \cref{lem:correctness-2}, we obtain that $\calI_1^{\log q}[W] = \OPT$ with probability at least $1 - 1/n^2$, which proves correctness.
    The running time is immediate from~\cref{lem:running-time}.
    Observe that we can obtain success probability $1 - 1/n^c$ for any constant $c\ge 2$ by repeating the algorithm $c/2$ times.
\end{proof}

\paragraph{Proof Sketch of \cref{thm:alg-knapsack-symmetric}}
Our presentation focused on proving \cref{thm:alg-knapsack}.
The proof of the symmetric variant stated in \cref{thm:alg-knapsack-symmetric} is very similar, thus we only sketch the required changes.
Essentially, we need to exchange profits with weights everywhere, which in turn means exchanging max-plus convolutions by min-plus convolutions.
In more detail:
Instead of working with the profit sequence $\Profit{\calI}$, we work with the weight sequence $\calW_{\calI}$,
where the entry $\calW_{\calI}[j]$ stores the minimum weight of a solution with profit at least $j$.
We do not know $\OPT$, but we can compute an approximation $\tilde{V}$ satisfying $\tilde{V} - \pmax \leq \OPT \leq \tilde{V}$ in linear time (see e.g.~\cite[Theorem 2.5.4]{KellererPP04}).
In the algorithm, we exchange all ocurrences of $\wmax$ by $\pmax$ and all ocurrences of $W$ by $\tilde{V}$.
With these changes, the functions $C_j^\ell$ are now $\wmax$-near convex (instead of $\pmax$-near concave) so we use \cref{thm:convex-minplus} directly instead of \cref{cor:concave-maxplus}.
In this way, we obtain the array $C^{\log q}_1\range{\tilde{V} - \pmax}{\tilde{V}} = \calW_{\calI}\range{\tilde{V} - \pmax}{\tilde{V}}$.
Then, we can infer $\OPT$ as the largest $i \in \range{\tilde{V}-\pmax}{\tilde{V}}$ such that $\calW_{\calI}[i] \leq W$.

\paragraph{Reconstructing an optimal solution}

So far we were only concerned with returning the optimal profit of a given Knapsack instance.
To reconstruct a solution $x \in \set{0,1}^n$ such that $p_{\calI}(x) = \OPT$, we proceed as follows.
After running~\cref{alg:knapsack}, we obtain the sequences $C^{\ell}_1[J^{\ell}]$ for every
$\ell \in \rangezero{\log q}$ and $j \in \range{1}{q/2^\ell}$.
For the output entry $C^{\log q}_1[W]$, we find a witness $i \in J^{\log q - 1}$, i.e., an index $i$ such that
$C_1^{\log q-1}[i] + C^{\log q - 1}_2[W - i] = C^{\log q}_1[W]$.
This can be done in time $|J^{\log q-1}| = \Ohtilde(\sqrt{\Delta q/2})$ by simply trying all possibilities.
Then, we continue recursively finding witnesses for $i$ and $W - i$.
Eventually, we reach one entry in each array $C_j^0$ for $j \in \range{1}{q}$.
Note that this takes time proportional to the length of all sequences
$\sum_{\ell=0}^{\log q}q/{2^\ell} \cdot O(|J^{\ell}|) = \Ohtilde(q \sqrt{\Delta}) = \Ohtilde(\sqrt{q \wmax W}) \leq \Ohtilde(n\wmax)$,
where the last step uses $q \leq W/\wmax$ and $W \leq n \wmax$.
Finally, observe each array $C_j^0$ was computed using the standard dynamic programming algorithm of \cref{fact:dp},
which allows to retrieve a solution for an fixed entry $C^0_j[i]$ in the same time it takes to compute it.
Thus, we can retrieve the optimal solution with no extra overhead on the overall running time.

%% file: sections/minplus-algorithm.tex
\section{MinPlus Convolution for Near-Convex Sequences}\label{sec:minplus-algorithm}

In this section we prove~\cref{thm:convex-minplus}.

\thmconvexminplus*

\subsection{Preparations}

Throughout this section, fix the functions $f \colon [n] \mapsto \range{-U}{U}, g \colon [m] \mapsto \range{-U}{U}$.
Recall that we say that $f \colon\rangezero{n} \mapsto \Int$ is $\Delta_f$-near convex, if there is a convex function $\cvx{f} \colon \rangezero{n} \mapsto \Rat$ such that $\cvx{f}(i) \leq f(i) \leq \cvx{f}(i) + \Delta_f$ for all $i \in \rangezero{n}$ (see~\cref{def:nearconvex}).
First observe that the lower convex hull of the points $\set{(i, f(i)) \mid i \in \rangezero{n}}$ gives the pointwise
maximal convex function $\cvx{f}$ with $\cvx{f} \leq f$.
This can be computed in time $O(n)$ by Graham's scan~\cite{Graham72}, since the points are already sorted by $x$-coordinate.
Then, we can infer $\Delta_f = \max\set{1, \, \max_{i \in \rangezero{n}}f(i) - \cvx{f}(i)}$.
Thus, from now on we assume that we know $\cvx{f}, \Delta_f, \cvx{g}, \Delta_g$. Set $\Delta := \max\set{\Delta_f, \Delta_g}$.
Let $\cvx{h} := \minconv{\cvx{f}}{\cvx{g}}$ and $h := \minconv{f}{g}$. The goal is to compute $h$.

We start by introducing some notation.
We call $(i, j) \in \rangezero{n} \times \rangezero{m}$ a \emph{point}.
We visualize a point $(i,j)$ as lying on the $i$-th row and $j$-th column of an $n \times m$ grid, where $(0,0)$ is on the bottom-left corner and $(n, m)$ on the top right corner.
A point $(i, j)$ lies on \emph{diagonal} $i + j$.
For any $\delta \geq 0$, a point $(i, j)$ is \emph{$\delta$-relevant} if $\cvx{f}(i) + \cvx{g}(j) \leq \cvx{h}(i + j) + \delta$.
We denote by $R_\delta$ the set of all $\delta$-relevant points.

Points that are 0-relevant are important because of the following observation:
We call $i$ a \emph{witness} for $\cvx{h}(k)$ if $\cvx{f}(i) + \cvx{g}(k-i) = \cvx{h}(k)$.
Thus, observe that $i$ is a witness for $\cvx{h}(k)$ if and only if $(i, k-i)$ is a 0-relevant point.

The importance of $2\Delta$-relevant points is captured by the following lemma:

\begin{lemma}\label{lem:relevant-points}
    If $(i, j) \notin R_{2\Delta}$ then $f(i) + g(j) > h(i + j)$.
\end{lemma}
That is, points that are not $2\Delta$-relevant can be ignored for the purpose of computing $h$.
\begin{proof}
    Since $(i, j)$ is not $2\Delta$-relevant, it holds that $f(i) + g(j) \geq \cvx{f}(i) + \cvx{g}(j) > \cvx{h}(i+j) + 2\Delta$.
    Let $k := i + j$, and let $i^*$ be a witness for $\cvx{h}(k)$, i.e., $\cvx{f}(i^*) + \cvx{g}(k-i^*) = \cvx{h}(k)$.
    Then,
    \[
      h(k) \leq f(i^*) + g(k-i^*) \leq \cvx{f}(i^*) + \Delta + \cvx{g}(k-i^*) + \Delta = \cvx{h}(k) + 2\Delta < f(i) + g(j). \qedhere
    \]
\end{proof}

We say that a set of points $P$ is a \emph{monotone path} if for every $k \in \rangezero{n+m}$ $P$ contains exactly one point $(i_k, j_k)$ on diagonal $k$,
and we have $(i_{k+1}, j_{k+1}) \in \set{(i_k +1, j_k), (i_k, j_k +1)}$ for every $k \in \rangezero{n+m-1}$, see~\cref{fig:monotone-path} for an illustration.
For any $\delta > 0$, we let
\begin{align*}
    P^+_\delta &:= \set{(i,k-i) \mid k \in \rangezero{n+m}, i \in \rangezero{n} \text{ is maximal s.t.}\ (i, k-i) \text{ is }\delta\text{-relevant}}, \\
    P^-_\delta &:= \set{(i,k-i) \mid k \in \rangezero{n+m}, i \in \rangezero{n} \text{ is minimal s.t.}\ (i, k-i) \text{ is }\delta\text{-relevant}}.
\end{align*}

The next two lemmas show that $P^+_\delta, P^-_\delta$ are monotone paths and that $P^+_\delta, P^-_\delta$ form the boundary of $R_\delta$, see~\cref{fig:delta-region} for an illustration.
This establishes structure of $R_\delta$ that we will be exploit later.

\begin{figure}[ht]
    \begin{subfigure}[c]{0.32\textwidth}
        \centering
        \includegraphics[width=\linewidth]{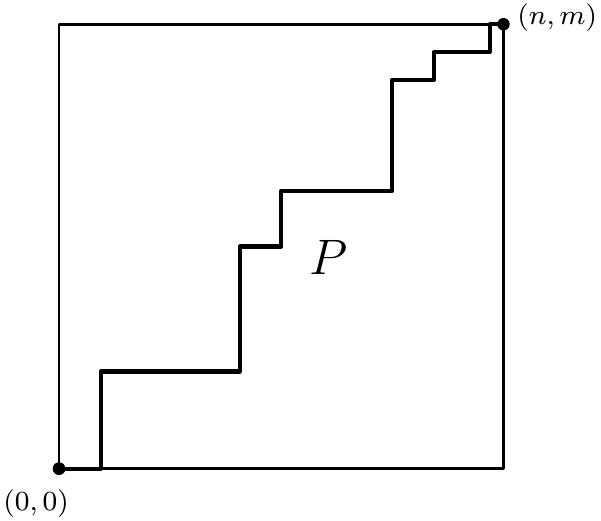}
        \caption{A monotone path $P$}
        \label{fig:monotone-path}
    \end{subfigure}
    \hfill
    \begin{subfigure}[c]{0.32\textwidth}
        \centering
        \includegraphics[width=\linewidth]{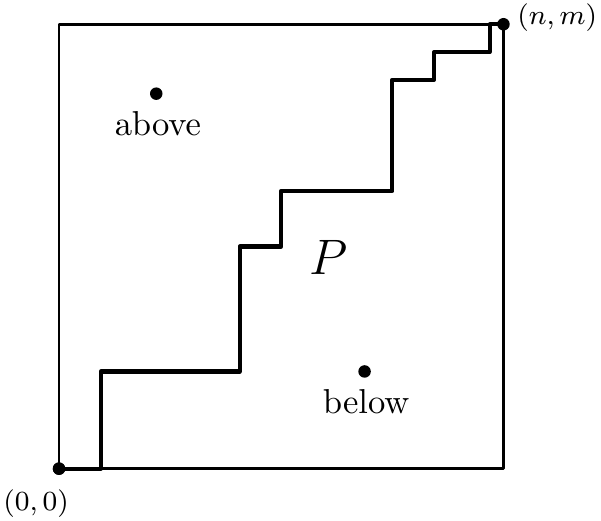}
        \caption{Points above and below $P$}
        \label{fig:above-below}
    \end{subfigure}
    \hfill
    \begin{subfigure}[c]{0.32\textwidth}
        \centering
        \includegraphics[width=\linewidth]{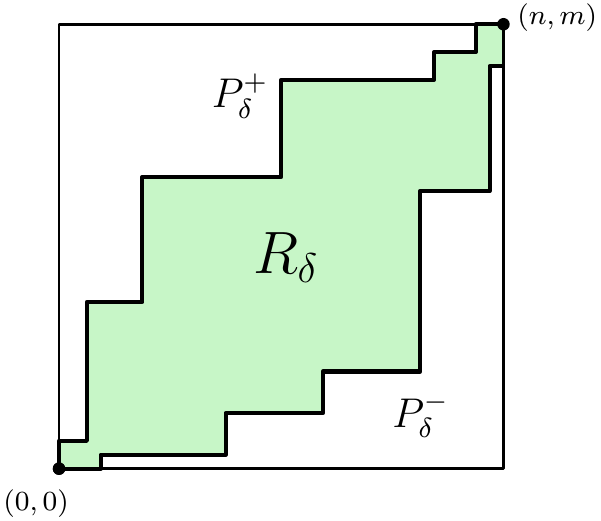}
        \caption{$R_\delta$ is between $P_\delta^+$ and $P_\delta^-$}
        \label{fig:delta-region}
    \end{subfigure}
    \caption{Visualizations for concepts used in~\cref{sec:minplus-algorithm}.}
    \label{fig:vis-paths}
 \end{figure}

\begin{restatable}[Monotone Paths]{lemma}{lemmonotonepaths}\label{lem:monotone-paths}
    For any $\delta \geq 0$, $P^-_\delta, P^+_\delta$ are monotone paths.
\end{restatable}
\begin{proof}
    Since for each $k \in \rangezero{n+m}$, $\cvx{h}(k)$ has a witness, there is a 0-relevant point $(i, k-i)$.
    Since every 0-relevant point is also $\delta$-relevant, it follows that $P^-_\delta$ contains exactly one point $(i_k, j_k)$ with $i_k + j_k = k$ for every $k \in \rangezero{n+m}$.

    In the following, we show that $i_{k-1} \leq i_k \leq i_{k-1}+1$ holds for all $k \in \range{1}{n+m}$.
    Since $i_k + j_k = k$, it then also follows that $k-1 - j_{k-1} \leq k - j_k \leq k-1-j_{k-1}+1$, which yields $j_{k-1} \leq j_k \leq j_{k-1} + 1$.
    Since $i_{k+1} + j_{k+1} = k + 1 = i_k + j_k + 1$, it follows that $(i_k, j_k) \in \set{(i_{k-1}+1,j_{k-1}), (i_{k-1}, j_{k-1}+1)}$.
    So it remains to prove $i_{k-1} \leq i_k \leq i_{k-1} + 1$.
    We distinguish two cases.

    \begin{description}
        \item[Case 1: $i_k \leq i_{k-1}$.] We show that in this case $i_k \geq i_{k-1}$ (and thus $i_k = i_{k-1}$).
        Let $i^*_{k-1}$ be a witness for $\cvx{h}(k-1)$.
        Note that $i^*_{k-1} \geq i_{k-1}$ by definition of $P^-_\delta$.
        We have
        \[
            \cvx{f}(i_k) + \cvx{g}(k - i_k) \leq \cvx{h}(k) + \delta \leq \cvx{f}(i^*_{k-1}) + \cvx{g}(k-i^*_{k-1})+\delta,
        \]
        where the first inequality follows due to the definition of $P_\delta^-$ and the second due to the definition of $\cvx{h}$.
        Rearranging, we get
        \begin{equation}
            \cvx{g}(k-i_k) - \cvx{g}(k-i^*_{k-1}) \leq \cvx{f}(i^*_{k-1}) - \cvx{f}(i_k) + \delta. \label{proof:monotone-paths:eqn:1}
        \end{equation}
        Since $i^*_{k-1} \geq i_{k-1} \geq i_k$, we have $k-i_k \geq k-i^*_{k-1}$.
        By convexity of $\cvx{g}$, we obtain
        \begin{equation}
            \cvx{g}(k-1-i_k) - \cvx{g}(k-1-i^*_{k-1}) \leq \cvx{g}(k-i_k) - \cvx{g}(k-i^*_{k-1}). \label{proof:monotone-paths:eqn:2}
        \end{equation}
        Combining \eqref{proof:monotone-paths:eqn:1} and \eqref{proof:monotone-paths:eqn:2} and rearranging, we obtain
        \[
            \cvx{f}(i_k) + \cvx{g}(k-1-i_k) \leq \cvx{f}(i_{k-1}^*) + \cvx{g}(k-1-i^*_{k-1}) + \delta = \cvx{h}(k-1) + \delta,
        \]
        where the last equality is by definition of $i^*_{k-1}$.
        Thus, $(i_k, k-1-i_k)$ is $\delta$-relevant, and since $i_{k-1}$ is minimal such that $(i_{k-1}, k-1-i_{k-1})$ is $\delta$-relevant we obtain $i_{k-1} \leq i_k$, as desired.

        \item[Case 2: $i_k > i_{k-1}$.] We show that in this case $i_k \leq i_{k-1} + 1$ (and thus, $i_k = i_{k-1} + 1)$.
        Let $i^*_k$ be a witness for $\cvx{h}(k)$.
        By definition of $P^-_\delta$, we have $i^*_k \geq i_k$. Moreover,
        \[
            \cvx{f}(i_{k-1}) + \cvx{g}(k-1-i_{k-1}) \leq \cvx{h}(k-1)+\delta \leq \cvx{f}(i_k^*-1) + \cvx{g}(k-i_k^*)+\delta,
        \]
        where the first inequality is due to the definition of $P^-_\delta$ and the second due to the definition of $\cvx{h}$.
        Rearranging, we get
        \begin{equation}
            \cvx{g}(k-1-i_{k-1}) - \cvx{g}(k-i^*_k) \leq \cvx{f}(i^*_k -1) - \cvx{f}(i_{k-1}) + \delta. \label{proof:monotone-paths:eqn:3}
        \end{equation}
        Since $i^*_k -1 \geq i_k - 1 \geq i_{k-1}$ and by the convexity of $\cvx{f}$, we have
        \begin{equation}
            \cvx{f}(i_k^* - 1) - \cvx{f}(i_{k-1}) \leq \cvx{f}(i^*_k) - \cvx{f}(i_{k-1}+1). \label{proof:monotone-paths:eqn:4}
        \end{equation}
        Combining and rearranging \eqref{proof:monotone-paths:eqn:3} and \eqref{proof:monotone-paths:eqn:4}, we obtain
        \[
            \cvx{f}(i_{k-1}+1) + \cvx{g}(k-1-i_{k+1}) \le \cvx{f}(i^*_k) - \cvx{g}(k-i^*_k) + \delta = \cvx{h}(k) + \delta,
        \]
        where the last equality holds by definition of $i^*_k$.
        Hence, $(i_{k-1}+1, k-1-i_{k+1})$ is $\delta$-relevant.
        Since its diagonal is $k$, $i_{k-1}+1$ is a possible choice for $i_k$.
        By minimality of $i_k$ (due to the definition of $P^-_\delta$), we obtain that $i_k \leq i_{k-1}+1$.
    \end{description}

    In both cases we obtain $i_k \in \set{i_{k-1}, i_{k-1}+1}$, proving the claim.
    This finishes the proof for $P^-_\delta$.
    The proof for $P^+_\delta$ is symmetric (replacing the roles of $f$ and $g$ essentially flips $P^-_\delta$ and $P^+_\delta$).
\end{proof}

Let $(i,j)$ be a point and $P$ a monotone path.
Let $(a,b) \in P$ be the unique point on the same diagonal as $(i,j)$, i.e., $a + b = i + j$.
We say that $(i, j)$ is \emph{below $P$} if $i < a$, \emph{above $P$} if $i > a$, and \emph{on $P$} if $i = a$, see~\cref{fig:above-below} for an illustration.

\begin{lemma}\label{lem:boundary-relevant}
    For any $\delta \geq 0$, $R_\delta$ consists of all points $(i, j)$ that are on or below $P^+_\delta$ and on or above $P^-_\delta$.
\end{lemma}
\begin{proof}
    Fix $k \in \rangezero{n + m}$ and let $(i^+, k-i^+), (i^-, k-i^-)$ be the point on diagonal $k$ in $P^+_\delta$ and $P^-_\delta$, respectively.
    Consider any $(i, j) \in R_\delta$ on diagonal $k$.
    By maximality of $i^+$ we have $i \leq i^+$, and similarly $i \geq i^-$ by the minimality of $i^-$.
    Thus, no point in $R_\delta$ is above $P^+_\delta$ or below $P^-_\delta$.
    It remains to show that for any $i^- \leq i \leq i^+$ we have $(i, k-i) \in R_\delta$.
    Note that the function $r(i) := \cvx{f}(i) + \cvx{g}(k-i)$ is convex (since it is the sum of convex functions).
    Since $(i^+, k-i^+)$ is $\delta$-relevant, we have $r(i^+) \leq \cvx{h}(k) + \delta$.
    Similarly, since $(i^-, k-i^-)$ is $\delta$-relevant, we have $r(i^-) \leq \cvx{h}(k) + \delta$.
    By convexity of $r$, we obtain that $r(i) \leq \cvx{h}(k) + \delta$ for all $i^- \leq i \leq i^+$.
    Hence, we conclude that for each $i^- \leq i \leq i^+$ we have $(i, k-i) \in R_\delta$.
\end{proof}

Finally, we need some background on \emph{sumsets}.
Given $A, B \subseteq \range{-U}{U}^2$ where $U \in \Nat$, we define $A + B = \set{a + b \mid a \in A, b \in B}$ as their sumset, where the addition $a+b$ is done componentwise.
The naive way to compute $A + B$ takes time $O(|A| \cdot |B|)$.
For our application, we want to compute the sumset in time near linear in its size $|A + B|$.
For this end, we will use the following tool to compute \emph{sparse non-negative convolution}.
Given vectors $P, Q \in \Nat^n$, their \emph{convolution} $P \conv Q \in \Nat^{2n-1}$ is defined coordinate-wise by
$(P \conv Q)[k] = \sum_{i+j=k}P[i] \cdot Q[j]$.

\begin{theorem}[{Deterministic Sparse Convolution~\cite{BringmannFN22}}]\label{thm:sparseconvolution}
    There is a deterministic algorithm to compute the convolution of two nonnegative vectors $A, B \in \Nat^n$
    in time $O(t\polylog(n\Delta))$, where $t$ is the number of non-zero entries in $A \conv B$ and $\Delta$ is the largest entry in $A$ and $B$.
\end{theorem}

See also~\cite{BringmannFN21} for improvements in the log-factors at the cost of randomization
and \cite{ColeH02,Nakos20,GiorgiGC20} for prior randomized algorithms with similar guarantees.

\begin{corollary}[Output Sensitive Sumset Computation]\label{cor:output-sensitive-sumset}
    Given $A, B \subseteq \range{-U}{U}^2$, with $|A + B| \leq N$, $A + B$ can be computed in time $\Ohtilde(N)$.
\end{corollary}
\begin{proof}
    Let $A' := \set{(x + U) \cdot 5U + (y + U) \mid (x, y) \in A}$ and similarly, let $B' := \set{(x + U) \cdot 5U + (y + U) \mid (x, y) \in B}$.
    Observe that this is a one-to-one embedding of $A, B \subseteq \range{-U}{U}^2$ into $A', B' \subseteq \rangezero{\Theta(U^2)}$.
    Moreover, one can check that given $C' := A' + B'$ we can infer $C := A + B$
    (the choice of $5U$ prevents any interactions between coordinates when summing them up).

    Thus, it suffices to compute $A' + B'$.
    To this end, construct their indicator vectors $P_{A'}, P_{B'} \in \Nat^{\Theta(U^2)}$ and compute the convolution
    $P_{C'} = P_{A'} \conv P_{B'}$. The non-zero entries in $P_{C'}$ correspond to the elements of $A' + B'$.
    By~\cref{thm:sparseconvolution}, this runs in time $O(|A'+B'|\polylog(N, U)) = \Ohtilde(N)$.
\end{proof}

\subsection{Algorithm}

We are ready to describe our algorithm. Recall that we have access to the functions $f, \cvx{f}, g, \cvx{g}$ and the value $\Delta = \max\set{\Delta_f, \Delta_g}$.

\paragraph{Computing $\bm{\cvx{h} = \minconv{\cvx{f}}{\cvx{g}}}$.}
Consider the pseudocode given in \cref{alg:cvx-minplus}.

\begin{algorithm}[ht]
    \caption{Given convex functions $\cvx{f} \colon \rangezero{n} \mapsto \Rat, \cvx{g} \colon \rangezero{m} \mapsto \Rat$,
    the algorithm computes $\cvx{h} = \minconv{\cvx{f}}{\cvx{g}}$.}\label{alg:cvx-minplus}
    \begin{algorithmic}[1]
    \State $i^*_0 \gets 0, \cvx{h}(0) \gets \cvx{f}(0) + \cvx{g}(0)$
    \For{$k = 1,\dots,n + m$}
        \State $i_k^* \gets \argmin\set{\cvx{f}(i) + \cvx{g}(k-i) + \tfrac{i}{2n} \mid i \in \set{i^*_{k-1}, i^*_{k-1}+1} \cap \rangezero{n}}$ \label{alg:cvx-minplus:line:witness}
        \State $\cvx{h}(k) \gets \cvx{f}(i_k^*) + \cvx{g}(k-i_k^*)$
    \EndFor
\end{algorithmic}
\end{algorithm}

\begin{lemma}
    \cref{alg:cvx-minplus} computes $\cvx{h} = \minconv{\cvx{f}}{\cvx{g}}$ in time $O(n + m)$.
\end{lemma}
\begin{proof}
    The running time is immediate.
    To see correctness, focus on $i^*_k$ for $k \in \rangezero{n+m}$ as computed in \cref{alg:cvx-minplus}.
    We claim that the path $P_0^-$ equals $\set{(i^*_k, k - i^*_k) \mid k \in \rangezero{n + m}}$.
    That is, we want to argue that  $i^*_k$ is the minimum witness of $\cvx{h}(k)$ for each $k \in \rangezero{n + m}$.
    Indeed, by \cref{lem:monotone-paths}, $P_0^-$ is a monotone path.
    Thus, $i_k^* \in \set{i^*_{k-1}, i^*_{k-1} + 1}$.
    Observe that in \cref{alg:cvx-minplus:line:witness} we pick $i_k^*$ as the minimizer of $\cvx{f}(i) + \cvx{g}(k-i) + \tfrac{i}{2n}$ where $i \in \set{i^*_{k-1}, i^*_{k-1} + 1}$.
    Therefore, the algorithm correctly computes $i^*_k$ (the additive term $i/(2n)$ ensures that we choose the minimal $i$).
    Since $i_k^*$ is a minimum witness of $\cvx{h}(k)$, the algorithm correctly computes $\cvx{h}(k)$ for all $k \in \rangezero{n+m}$.
\end{proof}

\paragraph{Computing $\bm{h = \minconv{f}{g}}$.}
Recall that $f \colon [n] \mapsto \Int$ and  $g \colon [m] \mapsto \Int$.
As a final simplification, we argue that we can assume without loss of generality that $n = m$, and $n + 1$ is a power of 2.
To this end, let $N$ be the smallest power of 2 greater than $\max\set{n, m}$.
We pad the functions to length $N$ by setting $f(n + j) := 2j \cdot W$ for $j \in \range{1}{N-1-n}$
and $g(m + j) := 2j\cdot W$ for $j \in \range{1}{N-1-m}$, where $W$ is an integer larger than $\max_{i \in \rangezero{n}}f(i) + \max_{j \in \rangezero{m}}g(j)$.
Observe that the entries $h(0),\dots,h(n+m)$ of the result $h = \minconv{f}{g}$ are unchanged (due to the choice of sufficiently large $W$), so we can read off the original result from the result of the padded functions.
Moreover, observe that the padding does not change the parameters $\Delta_f$ and $\Delta_g$.

\smallskip
Now we can describe the algorithm.
After running~\cref{alg:cvx-minplus} we can assume that we have computed $\cvx{h}$ and the witness path $P_0^- = \set{(i^*_k, k - i^*_k) \mid k \in \rangezero{n + m}}$.
We will make use of the following subroutines:

\begin{itemize}
    \item $\textsc{Relevant}(i,j)$: returns $\cvx{f}(i) + \cvx{g}(j) \leq \cvx{h}(i+j) + 2\Delta$.
    \item $\textsc{BelowWitnessPath}(i,j)$: returns $i < i^*_{i+j}$
    \item $\textsc{AboveWitnessPath}(i,j)$: returns $i > i^*_{i+j}$
\end{itemize}

Now we can compute $h = \minconv{f,g}$ by calling $\textsc{RecMinConv}(\range{0}{n}, \range{0}{m})$.
See~\cref{alg:nearcvx-minplus} for the pseudocode.

\begin{algorithm}[ht]
    \caption{Given intervals $I = [\,i_A\,.\,.\,i_B\,], J = [\,j_A\,.\,.\,j_B\,]$, the algorithm computes
the contribution of $f[I]$ and $g[J]$ to $\minconv{f}{g}$.}\label{alg:nearcvx-minplus}
    \begin{algorithmic}[1]

    \Procedure{RecMinConv}{$I = [\,i_A\,.\,.\,i_B\,], J = [\,j_A\,.\,.\,j_B\,]$}
        \If{$\textsc{AboveWitnessPath}(i_A, j_B)$ and $\textsc{NotRelevant}(i_A, j_B)$} \Comment{Case 1}
            \State \Return $\tilde h(k) = \infty$ for all $k \in \range{i_A + j_A}{i_B+j_B}$
        \EndIf
        \If{$\textsc{BelowWitnessPath}(i_A, j_B)$ and $\textsc{NotRelevant}(i_B, j_A)$} \Comment{Case 2}
            \State \Return $\tilde h(k) = \infty$ for all $k \in \range{i_A + j_A}{i_B+j_B}$
        \EndIf
        \If{$\textsc{Relevant}(i_A, j_B)$ and $\textsc{Relevant}(i_B, j_B)$} \Comment{Case 3}
            \State Compute $C \gets \set{(i, f(i)) \mid i \in I} + \set{(j, g(j)) \mid j \in J}$ using~\cref{cor:output-sensitive-sumset} \label{alg:nearcvx-minplus:line:sumset}
            \State Infer $\tilde{h}(k) \gets \min\set{y \mid (k,y) \in C}$ for all $k \in \range{i_A + j_A}{i_B + j_B}$
            \State \Return $\tilde{h}$
        \Else \Comment{Case 4}
            \State Split $I$ into two intervals $I_1, I_2$ of equal length, similarly split $J$ into $J_1, J_2$
            \State Recursively compute $\tilde{g}_{i,j} \gets \textsc{RecMinConv}(I_i, J_j)$ for $i,j \in \set{1,2}$
            \State \Return the pointwise minimum of the functions $\tilde{g}_{i,j}$ for $i,j \in \set{1,2}$
        \EndIf
    \EndProcedure
\end{algorithmic}
\end{algorithm}

\cref{alg:nearcvx-minplus} recursively computes the contribution of $f\range{i_A}{i_B}$ and $g\range{j_A}{j_B}$ to $h = \minconv{f}{g}$.
We next discuss its four cases; see~\cref{fig:cases} for illustrations of Cases 1-3.
If $(i_A, i_B)$ is above the witness path $P_0^-$ and is not $2\Delta$-relevant (Case 1), then as we argue below no point in $I \times J$ contributes to the output $h$, so in this case we return a dummy function (which is $+\infty$ everywhere).
Case 2 is symmetric, where $(i_B, j_A)$ is above $P_0^-$ and not $2\Delta$-relevant, and we again return a dummy function.
Case 3 applies when $(i_A, j_B)$ and $(i_B, j_A)$ are both $2\Delta$-relevant.
In this case, we explicitly compute $\tilde{h} = \minconv{f\range{i_A}{i_B}}{g\range{j_A}{j_B}}$
by computing the sumset $C = \set{(i, f(i)) \mid i \in I} + \set{(j, g(j)) \mid j \in J}$ and
infering $\tilde h(k)$ as the minimum $y$ such that $(k, y) \in C$, which by definition of the sumset equals the minimum $f(i) + g(j)$ such that $i \in I, j \in J$ and $i + j = k$.
Note that this step can be done for all $k \in \range{i_A + j_A}{i_B+j_B}$ in total time $O(|C|)$ by once scanning over all elements of $C$.

Finally, if none of the above cases apply, then we split both intervals $I$ and $J$ into equal halves and recurse on all
4 combinations of halves.
We combine them by taking the pointwise minimum of all computed functions.
This case is essentially brute force.

\begin{figure}[ht]
    \begin{subfigure}[c]{0.32\textwidth}
        \centering
        \includegraphics[width=\linewidth]{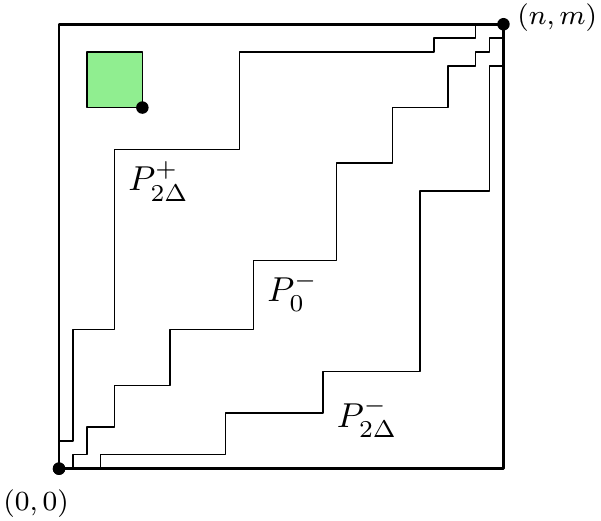}
        \caption{Case 1}
        \label{fig:subim1}
    \end{subfigure}
    \hfill
    \begin{subfigure}[c]{0.32\textwidth}
        \centering
        \includegraphics[width=\linewidth]{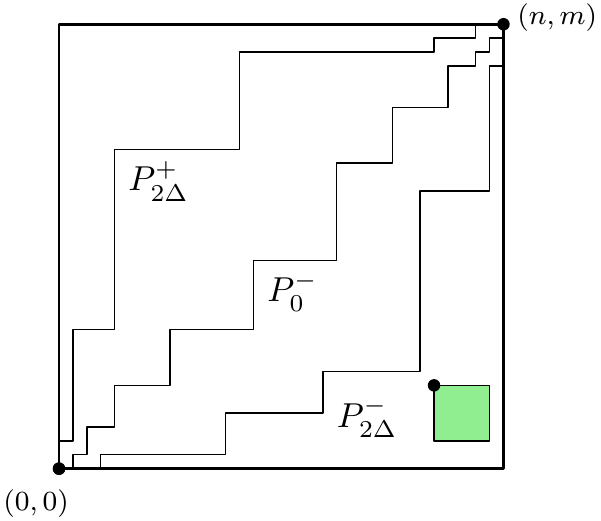}
        \caption{Case 2}
        \label{fig:subim2}
    \end{subfigure}
    \hfill
    \begin{subfigure}[c]{0.32\textwidth}
        \centering
        \includegraphics[width=\linewidth]{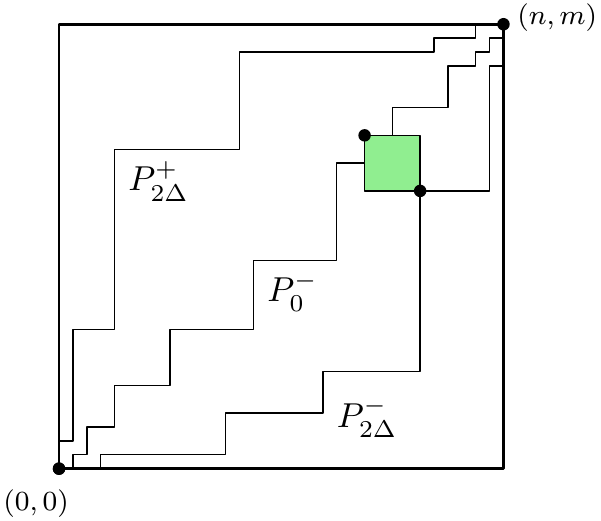}
        \caption{Case 3}
        \label{fig:subim3}
    \end{subfigure}
    \caption{Visualization of Cases 1-3 of \cref{alg:nearcvx-minplus}. The green box represents the current subproblem.}
    \label{fig:cases}
 \end{figure}

\subsection*{Correctness}
We start by analyzing the correctness of the algorithm.
\begin{lemma}[Correctness of \cref{alg:nearcvx-minplus}]
    $\textsc{RecMinConv}(\range{0}{n}, \range{0}{m})$ (\cref{alg:nearcvx-minplus}) correctly computes $h = \minconv{f}{g}$.
\end{lemma}
\begin{proof}
    Let $k \in \rangezero{n + m}$ and consider a point $(i^*, j^*)$ in diagonal $k$ such that $f(i^*) + g(j^*) = h(k)$, i.e., a witness for $h(k)$.
    We argue that some recursive call computes $f(i^*) + g(j^*)$.
    This is clear in Case 4, as $(i^*, j^*)$ is covered by one recursive subproblem.
    It is also clear in Case 3, since then $f(i^*) + g(j^*)$ is explicitly computed.

    To finish correctness, we argue that $(i^*, j^*)$ can never be in a subproblem to which Case 1 or 2 applies.
    Recall that Case 1 applies to a subproblem $I = \range{i_A}{i_B}, J = \range{j_A}{j_B}$ if $(i_A, j_B)$ is above $P^-_0$ and $(i_A, j_B)$ is not $2\Delta$-relevant.
    Since $(i_A, j_B)$ is not $2\Delta$-relevant, by \cref{lem:boundary-relevant} $(i_A, j_B)$ must be above $P^+_{2\Delta}$ or below $P^-_{2\Delta}$.
    Since $(i_A, j_B)$ is above $P^-_0$, it can only be above $P^+_{2\Delta}$.
    Since $(i_A, j_B)$ is the lower right corner of $I \times J$, it follows that all points in $I \times J$ are above $P^+_{2\Delta}$.
    Thus, by \cref{lem:boundary-relevant} all points in $I \times J$ are not $2\Delta$-relevant.
    If we assume for the sake of contradiction that $(i^*, j^*) \in I \times J$, then \cref{lem:relevant-points} implies $f(i^*) + g(j^*) > h(k)$, contradicting the choice of $(i^*, j^*)$ as a witness for $h(k)$.
    Hence, $(i^*, j^*)$ can never be in a Case 1 subproblem.
    Case 2 is symmetric. This finishes the correctness proof.
\end{proof}

\subsection*{Running Time}

Next, we analyze the running time.
The key insight is that in relevant regions both functions are essentially linear, with the same slope (see \cref{lem:same-slope}).
This implies that the sumset computed in Case 3 is small (see \cref{lem:small-sumset}), so it can be computed efficiently using~\cref{cor:output-sensitive-sumset}.
In the following two lemmas, let $I = \range{i_A}{i_B} \subseteq \rangezero{n}$ and $J = \range{j_A}{j_B} \subseteq \rangezero{m}$ be intervals of the same length $|I| = |J|$.

\begin{restatable}[Near Linearity inside Relevant Region]{lemma}{lemnearlinearity}\label{lem:same-slope}
    If $I \times J \subseteq R_{2\Delta}$ then there are $a,b,c \in \Real$ such that
    $|f(i) - (a \cdot i + b)| \leq 2\Delta$ for all $i \in I$ and $|g(j) - (a \cdot j + c)| \leq 2\Delta$ for all $j \in J$.
\end{restatable}
\begin{proof}
    Consider the linear interpolation between $(i_A, \cvx{f}(i_A))$ and $(i_B, \cvx{f}(i_B))$:
    \[
        F(x) := \frac{(i_B - x) \cvx{f}(i_A) + (x - i_A)\cvx{f}(i_B)}{i_B - i_A}.
    \]
    Similarly, consider \[G(x) := \frac{(j_B - x) \cvx{g}(j_A) + (x - j_A)\cvx{g}(j_B)}{j_B - j_A}.\]
    By convexity of $\cvx{f}$ and $\cvx{g}$, we have
    \begin{equation}\label{proof:same-slope:eqn:0}
        \cvx{f}(i) \leq F(i) \quad \forall i \in I, \quad  \quad \cvx{g}(j) \leq G(j) \quad \forall j \in J.
    \end{equation}
    Consider the diagonal $k := i_A + j_B$ and note that for all $i \in I$ we have $k - i \in J$ due to $|I| = |J|$.
    Thus, for each $i \in I$ the point $(i, k-i)$ is $2\Delta$-relevant, and we obtain
    \begin{equation}\label{proof:same-slope:eqn:1}
        \cvx{h}(k) \leq \cvx{f}(i) + \cvx{g}(k-i) \leq \cvx{h}(k) + 2\Delta \quad \forall i \in I
    \end{equation}
    This implies
    \begin{equation}\label{proof:same-slope:eqn:2}
        \cvx{h}(k) \leq F(i) + G(k-i) \leq \cvx{h}(k) + 2\Delta \quad \forall i \in I,
    \end{equation}
    since by \eqref{proof:same-slope:eqn:1} these inequalities hold for $i \in \set{i_A, i_B}$ and by the linear interpolation, they also hold in between.

    Now for any $i \in I$ we have
    \[
        \cvx{f}(i)
            \stackrel{\eqref{proof:same-slope:eqn:1}}{\geq} \cvx{h}(k) - \cvx{g}(k-i)
            \stackrel{\eqref{proof:same-slope:eqn:0}}{\geq} \cvx{h}(k) - G(k-i)
            \stackrel{\eqref{proof:same-slope:eqn:2}}{\geq} \cvx{h}(k) - (\cvx{h}(k) + 2\Delta - F(i)) = F(i) - 2\Delta.
    \]
    Thus, $\cvx{f}(i) \in \range{F(i) - 2\Delta}{F(i)}$, and by
    $\cvx{f} \leq f \leq \cvx{f} + \Delta_f \leq \cvx{f} + \Delta$, we obtain that $|f(i) - F(i)| \leq 2\Delta$.

    For $\cvx{g}(j)$ for any $j \in J$ we bound
    \[
        \cvx{g}(j)
            \stackrel{\eqref{proof:same-slope:eqn:1}}{\geq} \cvx{h}(k) - \cvx{f}(k-j)
            \stackrel{\eqref{proof:same-slope:eqn:0}}{\geq} \cvx{h}(k) - F(k-j),
    \]
    and
    \[
        \cvx{g}(j)
            \stackrel{\eqref{proof:same-slope:eqn:0}}{\leq} G(j)
            \stackrel{\eqref{proof:same-slope:eqn:2}}{\leq} \cvx{h}(k) - F(k-j) + 2\Delta.
    \]
    Therefore, $|\cvx{g}(j) - (\cvx{h}(k) - F(k-j) + \Delta)| \leq \Delta$.
    By linearity of $F$ we can write $F(k - j) = F(k) - F(j) + F(0)$.
    This yields $|\cvx{g}(j) - (F(j) + \lambda)| \leq \Delta$ for $\lambda := \cvx{h}(k) - F(k) - F(0) + \Delta$.
    Since $|g(j) - \cvx{g}(j)| \leq \Delta_g \leq \Delta$ we obtain $|g(j) - (F(j) + \lambda)| \leq 2\Delta$.
    Since $F$ is linear, writing $F(i) = a \cdot i + b$ and $F(j) + \lambda = a \cdot j + c$ finishes the proof.
\end{proof}

\begin{lemma}[Relevant Regions have Small Sumsets]\label{lem:small-sumset}
    If $I \times J \subseteq R_{2\Delta}$ then the sumset $\set{(i, f(i)) \mid i \in I} + \set{(j, f(j)) \mid j \in J}$ has size $O(\Delta \cdot (|I| + |J|))$.
\end{lemma}
\begin{proof}
    By \cref{lem:same-slope}, for any $(i,j) \in I\times J$ with $i + j = k$ we have
    \[ f(i) + g(j) = (a\cdot i + b) + (a \cdot j + c) \pm O(\Delta) = a\cdot k + b + c \pm O(\Delta).\]
    Thus, for each of the $|I| + |J| - 1$ $x$-coordinates (i.e., choices of $i + j$), there are $O(\Delta)$ different $y$-coordinates (i.e., values $f(i) + g(j)$) in the sumset.
\end{proof}

\begin{lemma}[Running Time of \cref{alg:nearcvx-minplus}]\label{lem:minplusrunningtime}
    $\textsc{RecMinConv}(\range{0}{n}, \range{0}{m})$ (\cref{alg:nearcvx-minplus}) runs in time $\Ohtilde(n\Delta)$.
\end{lemma}
\begin{proof}
    We first analyze the running time of one recursive subproblem, ignoring the cost of recursive calls.
    Note that in Cases 1 and 2 it suffices to return a dummy value, i.e., we do not need to iterate over $k \in \range{i_A+j_A}{i_B+j_B}$ to
    explitly return $\tilde{h}(k) = \infty$. Thus, Cases 1 and 2 run in time $O(1)$.
    We charge this time to the parent of the current subproblem, which is a Case 4-subproblem.

    Consider Case 4.
    Ignoring the cost of the recursive subproblems, Case 4 runs in time $O(1)$, which also covers the charging from children which fall in Cases 1 and 2.

    Consider Case 3, and let $s := i_B - i_A + 1 = j_B - j_A + 1$ be the current side length.
    By~\cref{lem:small-sumset}, the sumset computed in \cref{alg:nearcvx-minplus:line:sumset} has size $O(\Delta s)$.
    Thus, it can be computed in time $\Ohtilde(\Delta s)$ using \cref{cor:output-sensitive-sumset}, and the function $\tilde{h}$ can be inferred from it in time $O(\Delta s)$.

    Now we bound the total running time across subproblems.
    Fix a side length $s$ and consider all possible subproblems of side length $s$, i.e., all boxes
    \[
      B^s_{x,y} := \range{x \cdot s}{x \cdot s + s - 1} \times \range{y \cdot s}{y \cdot s + s - 1}, \text{ where } x,y \in \rangezero{n/s}.
    \]
    Consider a diagonal $D_{s,d} := \set{B^s_{x, x+d} \mid x \in \rangezero{n/s}}$ of these boxes, see \cref{fig:diagonal-boxes}.
    Note that a box in $D_{s,d}$ that lies fully above $P_{\Delta}^+$ corresponds to a Case 1-subproblem.
    A box in $D_{s,d}$ that lies fully below corresponds to a Case 2-subproblem.
    A box that is below or on $P_{2\Delta}^+$ and above or on $P_{2\Delta}^-$ corresponds to a Case 3-subproblem.
    The remaining boxes intersect $P_{2\Delta}^+$ or $P_{2\Delta}^-$ and correspond to Case 4.

    Note that by monotonicity of $P_{2\Delta}^+, P_{2\Delta}^-$, at most two boxes in $D_{s,d}$ are intersected by $P_{2\Delta}^+$ or $P_{2\Delta}^-$ and thus at most two boxes in $D_{s,d}$ can appear as Case 4-subproblems.
    Thus, Case 4 incurs time $O(1)$ per diagonal.
    We argue that among the boxes in $D_{s,d}$, at most two can appear as Case 3-subproblems.
    Indeed, if these would be at least three such boxes, then the parent of the middle box would also be between
    $P^+_{2\Delta}$ and $P^-_{2\Delta}$, and thus the parent would already be a Case 3-subproblem, see \cref{fig:3boxes,fig:parentbox}.
    Thus, the middle box would not get split, and it would not become a recursive subproblem.
    Hence per diagonal $D_{s,d}$, Case 3 incurs time $\Ohtilde(\Delta s)$ for each of at most two boxes.

    It remains to sum up over all side lengths $1 \leq s \leq n$ where $s = 2^\ell$ is a power of 2 (recall that at each recursive level we split the side length in two equal parts), and over all $O(n/s)$ diagonals $d$, to obtain total time
    $\sum_{\ell = 1}^{\log n} O(n/2^\ell) \cdot \Ohtilde(\Delta 2^\ell) = \Ohtilde(\Delta n)$.
    Note that the sum over $\ell$ only adds another log-factor, which is hidden by the $\Ohtilde$-notation.
\end{proof}

\begin{figure}[ht]
    \begin{subfigure}[c]{0.32\textwidth}
        \centering
        \includegraphics[width=\linewidth]{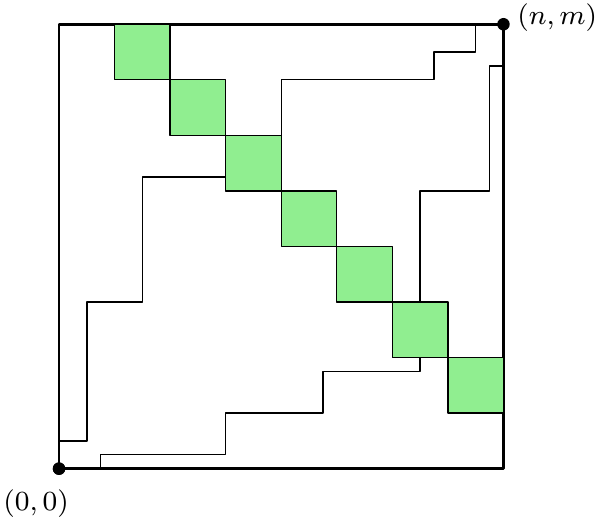}
        \caption{A diagonal of boxes $D_{s,d}$}
        \label{fig:diagonal-boxes}
    \end{subfigure}
    \hfill
    \begin{subfigure}[c]{0.32\textwidth}
        \centering
        \includegraphics[width=\linewidth]{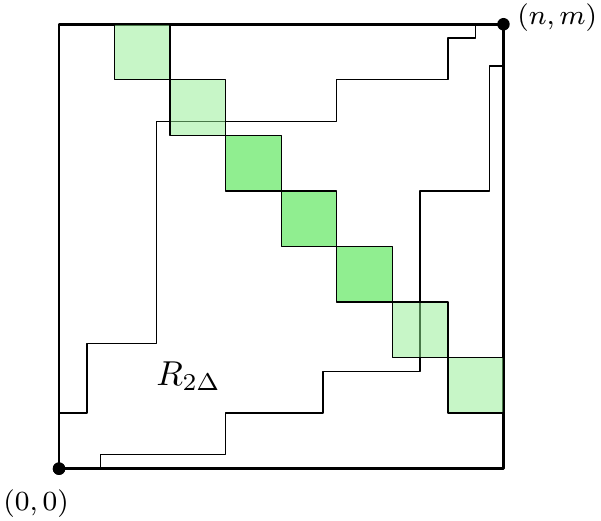}
        \caption{Three boxes inside $R_{2\Delta}$}
        \label{fig:3boxes}
    \end{subfigure}
    \hfill
    \begin{subfigure}[c]{0.32\textwidth}
        \centering
        \includegraphics[width=\linewidth]{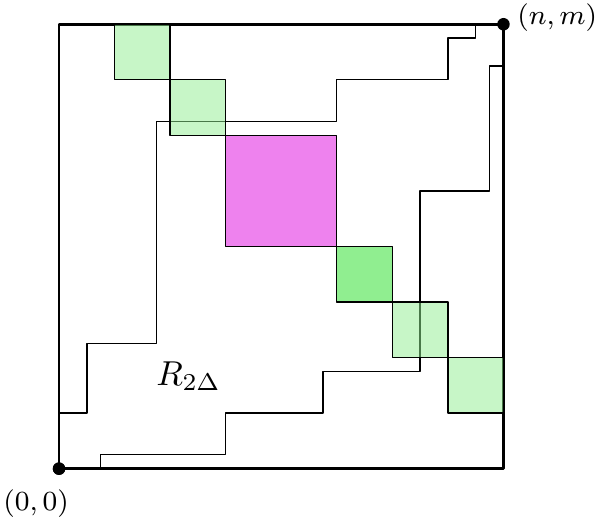}
        \caption{The parent box is already contained in $R_{2\Delta}$}
        \label{fig:parentbox}
    \end{subfigure}
    \caption{Visualizations for the proof of~\cref{lem:minplusrunningtime}.}
    \label{fig:boxes-runningtime}
\end{figure}